\documentclass[pra,twocolumn,citeautoscript,nofootbib]{revtex4-1}
\usepackage{amssymb, amsthm, amsmath, amsfonts, braket}
\usepackage{color, graphicx}
\usepackage{hyperref}
\usepackage{dcolumn}
\usepackage{bm}
\usepackage{subfigure}
\usepackage{amscd}
\usepackage{setspace}
\usepackage{enumerate}
\newcommand{\tb}[1]{\textbf{#1}}

\DeclareMathOperator{\Tr}{Tr}
\DeclareMathOperator{\Prob}{Prob}
\theoremstyle{plain}

\theoremstyle{plain}
\newtheorem{lemma}{Lemma}
\theoremstyle{plain}
 
\theoremstyle{plain}

\theoremstyle{remark}

\theoremstyle{conjecture}

\theoremstyle{observation}

\theoremstyle{definition}

\theoremstyle{corollary}

\theoremstyle{definition}

\theoremstyle{definition}

\theoremstyle{result}
\newtheorem{definition}{Definition}
\theoremstyle{assumption}

\theoremstyle{definition}

\theoremstyle{problem}

\theoremstyle{fact}

\newcommand*{\citen}[1]{%
  \begingroup
    \romannumeral-`\x 
    \setcitestyle{numbers}%
    \cite{#1}%
  \endgroup   
}

\begin{document}

\title{Topological Order and Memory Time in Marginally Self-Correcting Quantum Memory}
\author{Karthik Siva}\email{E-mail: ksiva@caltech.edu}
\author{Beni Yoshida}\email{E-mail: byoshida@perimeterinstitute.ca}
\affiliation{$^*$Institute for Quantum Information and Matter, California Institute of Technology, Pasadena, California 91125, USA}
\affiliation{$^\dagger$Perimeter Institute for Theoretical Physics, Waterloo, Ontario, N2L 2Y5 Canada}
\date{\today}

\begin{abstract}
We examine two proposals for marginally self-correcting quantum memory, the cubic code by Haah and the welded code by Michnicki. In particular, we prove explicitly that they are absent of topological order above zero temperature, as their Gibbs ensembles can be prepared via a short-depth quantum circuit from classical ensembles. Our proof technique naturally gives rise to the notion of free energy associated with excitations. Further, we develop a framework for an ergodic decomposition of Davies generators in CSS codes which enables formal reduction to simpler classical memory problems. We then show that memory time in the welded code is doubly exponential in inverse temperature via the Peierls argument. These results introduce further connections between thermal topological order and self-correction from the viewpoint of free energy and quantum circuit depth. 
\end{abstract}

\maketitle

\section{Introduction}

To create a usable memory system, classical or quantum, one must be able to write, store, and read out information reliably. The main difficulty in this endeavour is the protection of the information written against thermal noise during storage. Therefore, one natural question is the feasibility of self-correcting quantum memory, a device which would correct itself by dissipating energy to the thermal bath~\cite{Dennis02}. Instead of requiring active error-correction, the device would only require initialization of the system in a ground state and a final round of error-correction upon readout. Such a self-correcting quantum memory device could serve as the equivalent hard drive quantum computing technology. Here we define a self-correcting quantum memory as a system with local interactions, non-divergent terms, perturbative stability and diverging quantum memory time at nonzero temperature~\cite{Brell16}. See Ref.~\cite{Brown14_review, Landon-Cardinal15} for recent reviews.

The mechanism by which self-correction could be implemented is inspired by memory systems in classical physics. In classical memory, such as the two-dimensional Ising ferromagnet below a critical temperature, the energy cost associated with the erroneous flip dominates its entropic advantage so that it is likely to flip back and avoid proliferating. Thus a necessary ingredient for self-correcting quantum memory would be that the topologically ordered ground states of a quantum memory system are separated by large energy barriers. These ideas have inspired several proposals for self-correcting quantum memory~\cite{Haah11, Michnicki14, Bravyi13, Brell16, Kim11, Pedrocchi13, Bardyn16}. However, at this time, these are not verified or are not fully self-correcting in light of the aforementioned definition. We refer to these as ``marginally'' self-correcting. Additionally, several no-go results are known~\cite{Bravyi09, Haah10, Landon-Cardinal13, Pastawski15}. Namely, it has been shown that systems with both translation and scale invariance cannot be self-correcting in three dimensions~\cite{Beni11}.

Implementation of codewords as topologically ordered ground states of quantum many-body systems is especially attractive as codewords store logical qubits non-locally with stability against local perturbations. However, presence of topological order in the ground states is not sufficient for self-correcting behaviors. For example, the ground states of the two-dimensional toric code are topologically ordered as ground states are locally indistinguishable and cannot be prepared from product states via a short-depth circuit. Yet, due to point-like mobile excitations associated with string-like logical operators, it is not self-correcting~\cite{Alicki09, Komar16}. On the other hand, the four-dimensional toric code, which is provably a self-correcting quantum memory~\cite{Alicki10}, is also topologically ordered above zero temperature in a sense that its Gibbs ensemble cannot be prepared via a short-depth circuit~\cite{Hastings11}. This raises another important question concerning whether topological order can survive above zero temperature in three spatial dimensions, or in other words, whether Gibbs ensembles of three-dimensional physical systems can be prepared immediately, in light of the definition of thermal topological order proposed by Hastings~\cite{Hastings11}. It is currently not known whether aforementioned proposals of self-correcting quantum memories are topologically ordered at nonzero temperature or not. 

The question of topological order at nonzero temperature is expected to be equivalent to that of the feasibility of self-correcting quantum memory. However, the connection has been established only for simple models such as the two-dimensional and four-dimensional toric code. In this paper, we study static and dynamical properties of marginally self-correcting quantum memory models and further establish the connection between these two notions. We shall mainly focus on two specific examples, the cubic code and the welded code, but our techniques may extend to a larger class of systems. Our work also serves as a non-trivial test of the definition of thermal topological order proposed by Hastings.

\tb{Main results:}
The cubic code is a three-dimensional stabilizer code with fractal-shaped logical operator and $O(\log(n))$ energy barrier. We begin by proving that the cubic code, which breaks scale invariance, is not topologically ordered above zero temperature. More precisely, we show that the thermal Gibbs ensemble of the cubic code can be approximated from some classical ensemble via an $O(\log(n))$-range quantum circuit. (Note that non-constant range is necessary since the approximation error is assessed by trace norm). Further, this result extends to a larger class of quantum codes with fractal shaped logical operators. 

Although the result itself is perhaps not surprising, as evidences of thermal instability of the cubic code have been known, what may be new is the proof technique and also the physical intuition therein. First, we introduce the concept of an imperfect Hamiltonian $H(p)$ where interaction terms are removed from a stabilizer Hamiltonian $H$ with some probability $p$ which is set by the temperature: $p = \frac{2}{e^{\beta}+1}$. We then show that, if there (almost surely) exists a short-depth quantum circuit to disentangle $H(p)$, then the system is not topologically ordered at a given temperature. Second, we observe that the missing terms in an imperfect Hamiltonian $H(p)$ act as sinks for excitations so that by moving excitations to locations of missing terms, they can be eliminated from the system. This observation allows one to identify the Gibbs ensemble approximation problem as evaluations of free energies, enabling us to employ the Peierls argument to heuristically assess the absence and the presence of topological order at nonzero temperature. Finally, we prove that isolated excitations in the cubic code can be eliminated by moving them to nearby sinks. The proof utilizes polynomial representations of translation invariant stabilizer Hamiltonians in a crucial manner. 

Next, we turn our attention to the welded code which is free from string-like logical operators, has $O(n^{2/9})$ energy barrier and breaks translation invariance. We show that quantum memory time of the welded code is upper bounded by $O(\exp(\exp(c\beta)))$ for a fixed temperature. Namely, we show that quantum memory time of the welded code is upper bounded by classical memory time of the Ising model supported on a sparse lattice (Fig.~\ref{fig_sparse-lattice}). We then argue that memory time can be computed from the Peierls argument. Our derivation of memory time is not mathematically rigorous, and we back up our analytical result by numerical simulations of memory time of the sparse Ising model and confirm doubly exponential behavior. 

Again, although the result itself is perhaps not surprising, as relations between the one-dimensional Ising model and the two-dimensional toric code have been noticed, what may be new is the proof technique and also the physical intuition therein. The key idea is to choose a minimal set of Davies generators that ensure ergodicity in the weak-coupling limit since the gap of the Lindblad operator (in other words, the inverse of memory time) is lower bounded by that of a minimal choice. Thermal dynamics under minimal generators can be effectively described by a simpler classical system which enables a formal reduction to a classical problem. We will provide an algorithm to find such minimal generators which applies to arbitrary CSS stabilizer codes. This generalizes the work by Alicki \emph{et al} on Davies generators for the two-dimensional toric code to arbitrary CSS stabilizer codes~\cite{Alicki09}. Applications of this algorithm to the welded code enable reduction to problems concerning the Ising model on a sparse graph. We also reveal that the procedure of the welding has an interesting interpretation as attaching several copies of topological codes along the same gapped boundary. 

\tb{Plan of the paper:} The paper is organized as follows. In section~\ref{sec:topo}, we develop key theoretical tools to study topological order at nonzero temperature. In section~\ref{sec:cube}, we prove that the cubic code is not topologically ordered at nonzero temperature. In section~\ref{sec:decomposition}, we introduce an ergodic decomposition of Davies generators for CSS codes. In section~\ref{sec:weld}, we examine memory time of the welded code and provide numerical results. 

In appendix~\ref{sec:fractal_decomposition}, we present an ergodic decomposition for the cubic code and show that its memory time is upper bounded by classical memory time of two-dimensional classical fractal spin models. In appendix~\ref{sec:boundary}, we present an interpretation of the welding procedure from the viewpoint of anyon condensations on gapped boundaries. In appendix~\ref{sec:LGT}, we discuss generic properties of an imperfect Hamiltonian in three dimensions and show that a three-dimensional imperfect Hamiltonian can be coarse-grained into the form of lattice gauge theories. 

\section{Thermal topological order}\label{sec:topo}


Let us recall the definition of topological order at nonzero temperature due to Hastings~\cite{Hastings11}. Consider a frustration-free commuting Hamiltonian on a $D$-dimensional lattice of size $n=L^{D}$:
\begin{align}
H = - \sum_{j} Q_{j}
\end{align}
where $Q_{j}$ are projectors, $[Q_{i},Q_{j}]=0$ and ground states satisfy $Q_{j}|\psi\rangle=|\psi\rangle$ for all $j$. The Hamiltonian possesses topological order at zero temperature if it requires a unitary circuit with large circuit depth to create ground states. Namely, a ground state $|\psi\rangle$ is said to be $(R,\epsilon)$ trivial if there exists a unitary quantum circuit $U$ with range $R$ such that $| \psi - U \psi_\text{prod}|\leq \epsilon$ for some product state $\psi_\text{prod}$ where $|\ldots|$ represents the trace norm. By a range-$R$ unitary, we mean a unitary quantum circuit $U$ where the circuit depth multiplied by the maximum range of each unitary in the circuit is bounded by some range $R$~\cite{Chen10}. 

Topological order at nonzero temperature is a generalization of this notion to a thermal ensemble. Let $\rho_\beta$ be the Gibbs ensemble of the Hamiltonian $H$ at $\beta = T^{-1}$ with $k_{B}=1$; $\rho_{\beta}=\frac{e^{-\beta H}}{\text{Tr} e^{-\beta H}}$. The system is said to be topologically ordered at temperature $T$ if it requires a large depth quantum circuit to prepare the Gibbs ensemble from a classical ensemble. 

\begin{definition}[Hastings~\cite{Hastings11}]
Let each site $j$ have an additional degree of freedom $\mathcal{K}_j$, defining an enlarged Hilbert space $\mathcal{H}_{j}\otimes \mathcal{K}_{j}$ on each site. A density matrix $\rho$ is said to be $(R,\epsilon)$ trivial if there exists a unitary quantum circuit $U$ with range $R$ such that 
\begin{align}
|\rho_\beta - \Tr_{(\{\mathcal{K}_j\})}(U\rho_\text{cl}U^\dagger)| \leq \epsilon
\end{align}
where $\rho_{cl}$ is a classical state of range $R$, and where the trace is over the added degrees of freedom $\mathcal{K}_{j}$. A classical state of range $R$ is defined to be a thermal ensemble such that $\rho_{cl}=\mathcal{Z}^{-1}\exp(-H_{cl})$ where $H_{cl}$ is a Hamiltonian with range $R$ which is diagonal in a product basis. 
\end{definition}

We shall often focus on a special class of quantum memory models known as CSS stabilizer codes (Hamiltonians). A stabilizer code encodes $k$ logical qubits in the ground state subspace $\mathcal{L}$ of $n$ physical qubits supporting a stabilizer group $\mathcal{S}$. States spanning $\mathcal{L}$ are defined by their invariance under the action of elements of $\mathcal{S}$:
\begin{equation}
\mathcal{L} = \{\ket{\psi}\in (\mathbb{C}^2)^{\otimes n} : s\ket{\psi}=\ket{\psi} \forall s\in\mathcal{S}\}
\end{equation}
$\mathcal{S}$ is an Abelian subgroup of the Pauli group acting on $n$ qubits $\mathcal{P}_n$ defined as
\begin{equation}
\mathcal{P}_n = \langle iI, X_1, Z_1, X_2, Z_2,\dots X_n, Z_n\rangle
\end{equation}
The invariance of $\ket{\psi}$ under the action of $\mathcal{S}$ implies $-I\not\in\mathcal{S}$. Logical operators of the code are given by Pauli operators $P$ such that for the centralizer $\mathcal{C}(\mathcal{S})$, $P  = \mathcal{C}(\mathcal{S})\backslash \mathcal{S}$.
Additionally, one can always choose generators $S_i$ of $\mathcal{S}$ so that $\mathcal{S} = \langle S_1, S_2,\dots S_N\rangle$, and $\mathcal{L}$ is the space of ground states of a Hamiltonian $H = -\sum_{i=1}^N S_i$ with $N \geq n-k$. A CSS code is a stabilizer code for which every generator of $\mathcal{S}$ can be written as a product of $X$-type or $Z$-type stabilizer elements. The corresponding Hamiltonian can take the following special form
\begin{equation}
H = -\sum_{i=1}^{N_X} S^{(X)}_i - \sum_{i=1}^{N_Z} S^{(Z)}_i
\end{equation}
where $S^{(X)}_i$ is a product of solely Pauli $X$ operators and $S^{(Z)}_i$ is a product of solely Pauli $Z$ operators.

\subsection{Imperfect Hamiltonian}\label{sec:topo2}

To study the presence (or the absence) of topological order at nonzero temperature, it is convenient to consider an \emph{imperfect Hamiltonian}. Consider a stabilizer Hamiltonian $H = - \sum_{j=1}^{N} \frac{1+S_{j}}{2}$ where $S_{j}$ are Pauli stabilizer generators. 

\begin{definition}
Let $k_{j}=0,1$ be binary integers for $j=1,\ldots,N$ and define the imperfect Hamiltonian
\begin{align}
H(\vec{k}) = - \sum_{j=1}^{N} k_{j}S_{j}.
\end{align}
The maximally mixed ensemble in the ground state space of $H(\vec{k} )$ is denoted by $\rho(\vec{k})$. Define the free ensemble $\rho_{\text{free}}(\beta)$ by
\begin{equation}
\rho_\text{free}(\beta)=\sum_{\vec{k}} \Prob(\vec{k})\cdot \rho(\vec{k})
\end{equation}
where
\begin{equation}
\Prob(\vec{k}) = \prod_{j} (1-p)^{k_{j}}p^{(1-k_{j})}\qquad p = \frac{2}{e^{\beta}+1}.
\end{equation}
\end{definition}

In other words, we remove interaction terms from the stabilizer Hamiltonian with some probability $p$. In the next subsection, we will discuss the physical intuitions of the removed terms. 

The thermal ensemble $\rho_{\beta}$ can be approximated by the free ensemble $\rho_{\text{free}}(\beta)$ under a certain condition. In general, stabilizer generators $S_{j}$ are not independent since an identity operator can be formed by multiplying $S_{j}$. Let us introduce $N$-component binary vectors, denoted by $\vec{m}^{(i)}$ for $i=1,\ldots,2^{M}$, such that 
\begin{align}
\prod_{j=1}^{N} (S_{j})^{m^{(i)}_{j}}=I.\label{eq:constraint}
\end{align} 
Here we take $\vec{m}^{(1)}$ to be $m^{(1)}_{j}=0$ for all $j$. Note $k=n-(N-M)$.

\begin{definition}
Let $w(\vec{m}^{(i)})$ be the Hamming weight (the number of nonzero entries) of a binary vector $\vec{m}^{(i)}$. A stabilizer Hamiltonian is said to be $(\alpha,\gamma)$-free if there exist positive constants $\alpha > \gamma>0$ such that
\begin{align}
w(\vec{m}^{(i)}) \geq O(L^{\alpha}) \qquad i\not=1
\end{align}
and $M\leq O(L^{\gamma})$. 
\end{definition}

In computing $\rho_{\text{free}}(\beta)$, each stabilizer generator $S_{j}$ was removed randomly with probability $p(\beta)$ as if $S_{j}$ were independent operators. However, Eq.~(\ref{eq:constraint}) implies that $S_{j}$ are correlated with each other. The $(\alpha,\gamma)$-free condition assures that correlations are weak and $S_{j}$ can be treated as independent operators. Note that the cubic code (and quantum fractal codes) are $(2,1)$-free and the two-dimensional toric code is $(2,0)$-free.

\begin{lemma}\label{lemma_remove}
Let $H$ be an $(\alpha,\gamma)$-free stabilizer Hamiltonian. Then
\begin{align}
|\rho(\beta) - \rho_{\text{free}}(\beta)| \leq c\cdot \exp(-c'\cdot L^{\alpha})
\end{align}
where $c,c'>0$ are some constants independent of $\beta,L$.
\end{lemma}

The above lemma can be generalized to stabilizer Hamiltonians with qudits (spins with $d$ internal states). Namely, for a Hamiltonian $H = -\sum_{j}Q_{j}$ with $Q_{j}=\frac{1}{d}(I + S_{j}+ S_{j}^{2}+\cdots S_{j}^{d-1})$, we set the probability function by $p=\frac{d}{e^{\beta}-1+d}$.

\begin{proof}
The Gibbs ensemble is proportional to 
\begin{equation}
\begin{split}
\frac{e^{-\beta H}}{(e^{\beta}+1)^{N}} &= \prod_{j=1}^{N}\left( (1-p(\beta)) \cdot \frac{S_{j}+I}{2} + p(\beta) \cdot \frac{I}{2} \right)\\
&= \sum_{\vec{k}} \mbox{Prob}(\vec{k})\cdot \mathcal{M}(\vec{k})
\end{split}
\end{equation}
with $\mathcal{M}(\vec{k})=\prod_{j=1}^{N}\left( \frac{S_{j}+I}{2}\right)^{k_{j}}  \left(\frac{I}{2} \right)^{1-k_{j}}$. Here $\mathcal{M}(\vec{k})$ is proportional to $\rho(\vec{k})$ and can be written as $\mathcal{M}(\vec{k})= F(\vec{k}) \rho(\vec{k})$ where $F(\vec{k})= \Tr \mathcal{M}(\vec{k})$. Consider an $N$-component binary vector $\vec{k'}$. We say that $\vec{k'}$ is included in $\vec{k}$ if $k_{j}'\leq k_{j}$ for all $j$, and write $\vec{k'}\subseteq \vec{k}$. Then $\mathcal{M}(\vec{k})$ can be written as
\begin{align}
\mathcal{M}(\vec{k})= \frac{1}{2^{N}}\sum_{ \vec{k'}\subseteq \vec{k} } \prod_{j}
(S_{j})^{k_{j}'}.
\end{align}

Non-zero contributions to $\mbox{Tr}\ \mathcal{M}(\vec{k})$ come from $\vec{k'}$ such that $\prod_{j} (S_{j})^{k_{j}'}=I$ since the trace of non-identity Pauli operators is zero. So, $F(\vec{k})$ is equal to the number of $\vec{k'}\subseteq \vec{k}$ such that $\prod_{j} (S_{j})^{k_{j}'}=I$. Then we have
\begin{align}
\rho_{\beta} \propto \sum_{i=1}^{2^{M}} \sum_{\vec{k} \supseteq \vec{m}^{(i)}}\mbox{Prob}(\vec{k})\cdot \rho(\vec{k}).
\end{align}
We split the above sum into two parts
\begin{align}
\sum_{ \vec{k}}\mbox{Prob}(\vec{k})\cdot \rho(\vec{k}) + \sum_{i=2}^{2^{M}} \sum_{\vec{m}^{(i)} \subseteq \vec{k}}\mbox{Prob}(\vec{k})\cdot \rho(\vec{k}).\label{eq:decomposition}
\end{align}
Since the Hamming weight of $\vec{m}^{(i)}$ for $i>1$ is at least $L^{\alpha}$, we have 
\begin{align}
\sum_{i=2}^{2^M} \sum_{\vec{k} \supseteq \vec{m}^{(i)}}\mbox{Prob}(\vec{k})
\leq  2^{M}\cdot (1-p(\beta))^{L^{\alpha}}
\end{align}
which approaches zero exponentially since $\alpha>\gamma$. The first term in Eq.~(\ref{eq:decomposition}) is $\rho_{\text{free}}(\beta)$, so $\rho_{\beta}\approx \rho_{\text{free}}(\beta)$. This completes the proof.
\end{proof}

\subsection{Cleaning excitations}\label{sec:topo3}

Consider a CSS stabilizer Hamiltonian: $H = -\sum_{j} S_{j}^{(X)} -\sum_{j} S_{j}^{(Z)}$. We hope to transform this Hamiltonian into a classical Hamiltonian with Pauli-$Z$ terms only. A $X$-type stabilizer generator $S_{j}^{(X)}$ is said to be $R$-removable if there exists a Pauli $Z$ operator $V_{j}$ which is exclusively supported inside a box of linear length $R$ surrounding $S_{j}^{(X)}$ and satisfies
\begin{align}
[S_{j}^{(X)},V_{i}]=0\quad (i\not=j) \qquad \{S_{j}^{(X)},V_{j}\}=0.
\end{align}
The operator $V_{j}$ can eliminate an excitation associated with $S_{j}^{(X)}$ since $V_{j}$ can flip the sign of $S_{j}^{(X)}$ without affecting other stabilizer generators.  

\begin{lemma}
If all the $X$-terms $S_{j}^{(X)}$ in a CSS-type Hamiltonian $H$ are $R$-removable, then there exists a range $R$ unitary $U$ such that $UHU^{\dagger}$ is a classical Hamiltonian. 
\end{lemma}

\begin{proof}
Let $V_{j}$ be the removers for $S_{j}^{(X)}$. Consider the following range-$R$ unitary transformation
\begin{align}
U = \prod_{j} \exp\left(\frac{\pi}{4} S_{j}^{(X)}V_{j}\right). 
\end{align} 
Since $S_{j}^{(Z)}\rightarrow S_{j}^{(Z)}$ and $S_{j}^{(X)}\rightarrow V_{j}$, one has
\begin{align}
UHU^{\dagger} = -\sum_{j}S_{j}^{(Z)} - \sum_{j}V_{j}
\end{align}
which is a classical Hamiltonian. 
\end{proof}

Intuitively, the lemma says that if one can eliminate $S_{j}^{(X)}$-excitations by unitary operations acting inside a box of size $R$, then the Hamiltonian can be shown to be $R$-trivial. 



Let us show that two-dimensional toric code is not topologically ordered at nonzero temperature by using this lemma. This statement has been already proven by Hastings, but we provide a slightly different proof. Assume that a vertex operator $A_{v}$ at a vertex $v$ is missing in $H(p)$. Then the missing vertex $v$ behaves as a sink of charge excitations. Namely, for a single isolated charge at arbitrary location, let us move it to the location of missing vertex by applying some appropriate string-like Pauli $Z$ operators. Once it reaches the missing vertex $v$, the charge excitation can disappear. Therefore, vertices with missing vertex terms can absorb charge excitations.

Imagine that we split the entire lattice into a grid of $\sqrt{c\log(L)}\times \sqrt{c\log(L)}$ qubits. Then, for sufficiently large but finite $c$, there almost surely exists at least one missing vertex term per grid as in (Fig.~\ref{fig_sink}). Namely, the probability of having no missing vertex in a grid is $(1-p)^{c\log(L)} = L^{c\log(1-p)}$ which is polynomially vanishing since $\log(1-p)<0$. Thus, one can bring charge excitations in each grid to the sink within the same grid. Let $g$ denote a grid and $n(g)$ denote the number of sinks in $g$. Since there are in total $\frac{L^{2}}{c\log(L)}$ grids, the probability of having at least one missing vertex in each grid $g$ is given as follows:
\begin{equation}
\Prob[n(g) > 0\quad \forall g]  = (1- L^{c\log(1-p)})^{L^2/c\log(L)} \label{eq:toriclb}
\end{equation}
The RHS of Eq.~(\ref{eq:toriclb}) is lower bounded by $1 - \frac{L^{2}}{c\log(L)} L^{c\log(1-p)}$, so if $c > - \frac{2}{\log(1-p)}$, there exists at least one missing vertex in each grid with probability asymptotically approaching to unity. Therefore, the Gibbs ensemble $\rho_{\beta}$ is $(R,\epsilon)$-trivial for $R=\sqrt{c\log(L)}$ and $\epsilon=O(L^{-\delta})$ with some $\delta>0$ in the thermodynamic limit. In other words, one can approximate the Gibbs ensemble via $\sqrt{c\log(L)}$-range circuit. 

By a similar argument, one is able to observe that topologically ordered models with point-like excitations do not retain their topological order at any finite temperature. Such systems include the Levin-Wen model and the Walker-Wang model. In the case of the welded code,\cite{Michnicki14} bifurcations of otherwise point-like excitations are avoided in the presence of missing terms by the fact that the length of solids within the code grow as $O(L^{2/3})$, so by similar argument excitations can always be eliminated within a solid. See appendix~\ref{sec:boundary} for details.


\begin{figure}[htb!]
\centering
\includegraphics[width=0.65\linewidth]{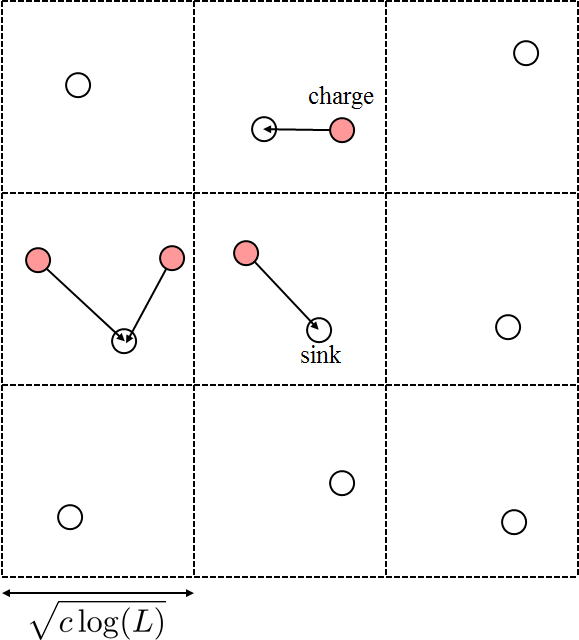}
\caption{White dots represent missing vertex terms which act as sinks of excitations. Red dots represents charge excitations that can be absorbed into a sink. The entire system is split into a grid of $\sqrt{c\log(L)}\times \sqrt{c\log(L)}$ qubits such that there is at least one sink per patch.
} 
\label{fig_sink}
\end{figure}

\subsection{Free energy}\label{sec:topo4}

\begin{figure}[htb!]
\centering
\includegraphics[scale=0.42]{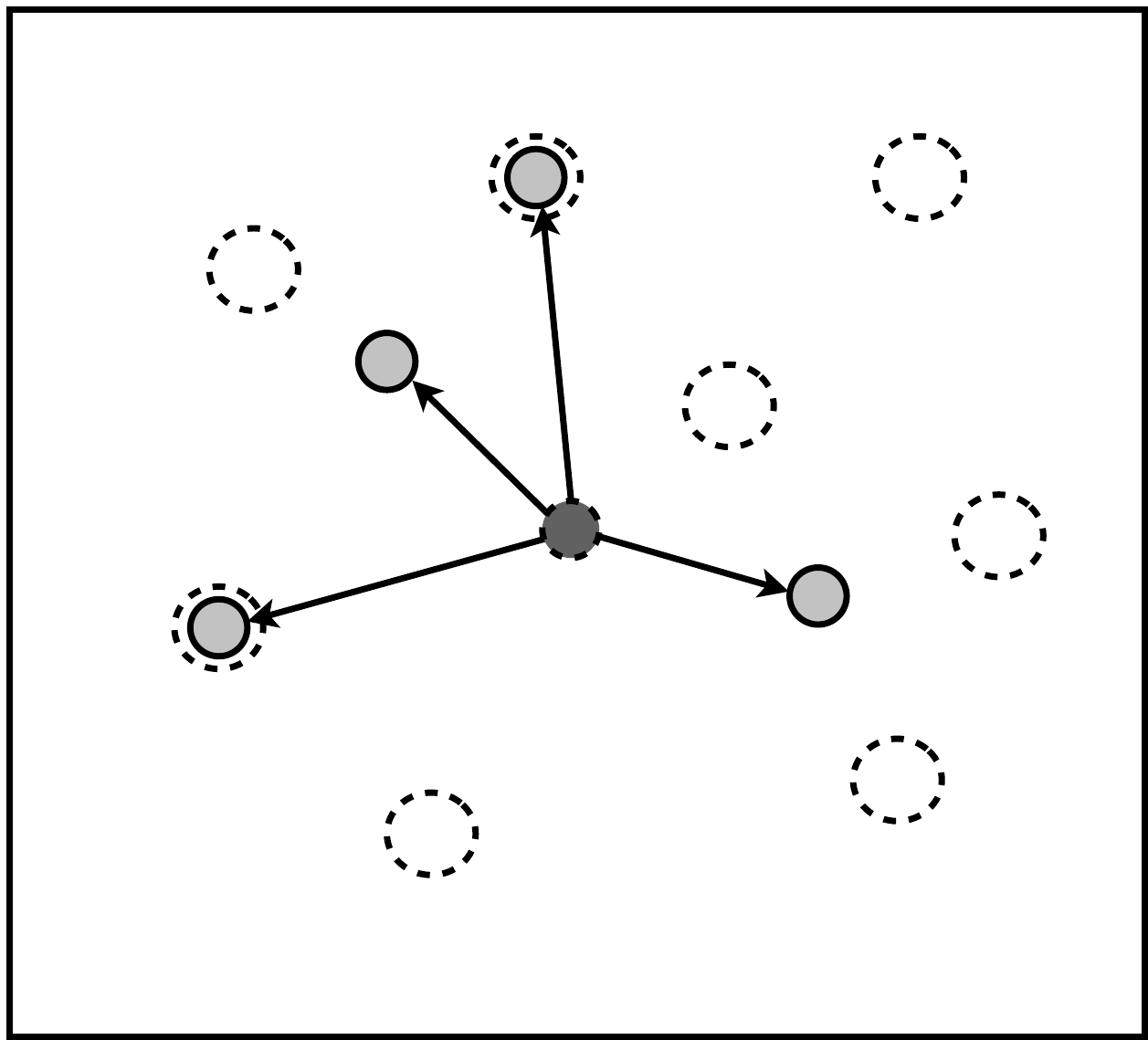}
\caption{An elimination of excitations and the Peierls argument. Grey circles represent excitations that evolved from the isolated single charge at the center. Dotted circles represent sinks of excitations. In this example, two excitations surrounded by sinks can be absorbed but the other two remain. 
} 
\label{fig_Peierls}
\end{figure}

Let us take a short detour and present physical interpretations of the arguments in previous sections. Consider an isolated single excitation which is located inside some block of finite size in the system and randomly distributed sinks of excitations resulting from an imperfect Hamiltonian $H(p)$. If the single excitation can be eliminated by moving it to sinks via unitary transformation acting strictly inside a box, one can construct a local unitary transformation which brings $H(p)$ into a classical Hamiltonian. For simplicity of discussion, we assume that the sinks appear at some finite constant rate $p_{0}>0$. If excitations are particle-like, meaning that their propagation trajectories are deformable lines, then an isolated excitation can be absorbed into some sink after propagating $1/p_{0}$ sites on average. However, when the system has diverging energy barrier and an excitation may split into multiple excitations, as in the cubic code and the welded code, it is not clear if excitations can be absorbed by randomly distributed sinks or not. 

Consider a single isolated excitation and a surrounding box of finite size. Let $\mathcal{R}$ be a distribution of sinks inside the box. We consider a set of all of the possible error configurations $\{\mathcal{C}_{j}\}$ into which the single isolated excitation may evolve. If there exists an excitation configuration $\mathcal{C}$ such that $\mathcal{C} \subseteq \mathcal{R}$, then the isolated excitation can be absorbed via a unitary transformation acting strictly inside the box. Let us consider an error configuration $\mathcal{C}$ such that the total number of excitations is given by $|\mathcal{C}|=E$. The probability for such a configuration that can be absorbed by $\mathcal{R}$ is given by
\begin{align}
{p_{0}}^{E} \sim \exp(-\beta_{0} E)\qquad e^{-\beta_{0}}\equiv p_{0}
\end{align}
We would like to estimate a probability that such an excitation configuration exist from a list of all the possible excitations configurations $\{\mathcal{C}_{j}\}$. While an exact calculation of such a probability is difficult in general, we can crudely approximate it by assuming that the probabilities for excitation configurations to be absorbed by $\mathcal{R}$ are statistically independent events. Let $N_{E}$ be the total number of error configurations such that $|\mathcal{C}|=E$. Then, the total probability for excitation configurations of size $E$ to be absorbed by sinks can be estimated as 
\begin{align}
P_\text{total}(E)\approx N_{E} \exp(-\beta_{0} E) = \exp ( - \beta_{0} F )
\end{align}
where 
\begin{align}
F = E-T_{0}S \qquad S=\log N_{E}. 
\end{align}
See Fig.~\ref{fig_Peierls}. Here we have approximated $P_\text{total}(E)$ by summing each probability $\exp(-\beta_{0} E)$ which is valid when $P_\text{total}(E)$ is small. The exponent $F$ can be viewed as a free energy of excitation configurations with energy $E$, and the approximation is valid for $F>0$.   

The observation above allows us to obtain an easy check for topological order at finite temperature by computing the free energy of excitations. If $F$ is negative for a given $E$, such excitations are likely to be absorbed. If $F$ is positive, such excitations are not likely to be absorbed. This is indeed the essence of the Peierls argument\cite{Peierls36}. Let us apply this argument to two-dimensional Ising model on a square lattice. 
The excitation energy $E$ corresponds to the length of the domain wall of an excitation droplet. The total number of possible configurations can be estimated to be $N_{E}=3^{E}$ by treating it as a random walk. (There are three directions to move on a square lattice). So the free energy of excitation at $E$ is lower bounded by
\begin{align}
F \geq E - \log 3^{E} T. \label{eq:freeenergy}
\end{align}
For $T$ less than some constant, the free energy is positive, and we expect that excitations cannot be absorbed by sinks. However, for large $T$, the free energy becomes negative, implying that such excitations can readily be absorbed. By finding the temperature for which the RHS becomes zero, we are able to obtain an analytical lower bound the transition temperature: $1/\log 3\sim 0.91 \leq T_{c}$.

\section{Cubic code}\label{sec:cube}

In this section, we show that the cubic code is not topologically ordered at nonzero temperature. 

\subsection{Fractal and polynomial}

This subsection provides a quick review of the cubic code and its generalization, the quantum fractal codes. To treat stabilizer quantum codes with fractal-shaped logical operators, it is convenient to introduce the polynomial representation of Pauli operators. See Ref.~\citen{Haah13,Beni13} for detailed discussions. Consider a three-dimensional cubic lattice with $L\times L \times L$ qubits. Each site is labeled as $(i,j,k)$ where $0\leq i,j,k\leq L-1$. Consider a polynomial over $\mathbb{F}_{2}$ fields
\begin{align}
f(x,y,z)= \sum_{i,j,k=0}^{L-1}c_{i,j,k}x^{i}y^{j}z^{k}\qquad c_{i,j,k}=0,1
\end{align}
and corresponding Pauli operators 
\begin{align}
X(f)=\prod_{i,j,k}{(X_{i,j,k})}^{c_{i,j,k}} \quad Z(f)=\prod_{i,j,k}{(Z_{i,j,k})}^{c_{i,j,k}}.
\end{align}
Polynomial representations are useful for dealing with translation invariant systems. For instance, $X(x^{i}f)$ represents a translation of $X(f)$ in the $\hat{x}$ direction by $i$ sites. We impose periodic boundary conditions by setting $x^{L}=y^{L}=z^{L}=1$. To be more concrete, we have $x^{-1}=x^{L-1}$, which implies that $-1$ th site and $L-1$ th site are identical. For simplicity of discussion, we assume $L=2^{m}$ with positive integer $m$. (The system size $L$ affects the number of logical qubits, but discussion on thermal topological order does not depend on choices of $L$).

Quantum fractal codes are defined on a cubic lattice where \emph{two qubits} live at each site of an $L\times L \times L$ cubic lattice. Pauli operators can be represented as $X \left(
\begin{array}{c}
f \\
g
\end{array}
\right)$ and $Z \left(
\begin{array}{c}
f \\
g
\end{array}
\right)$ where polynomials on the upper/lower row represent the first/second qubit at each site. 

Let us now specify the Hamiltonian of quantum fractal codes. For arbitrary polynomials $\alpha(x,y,z)$ and $\beta(x,y,z)$, consider
\begin{align}
Z \left(
\begin{array}{c}
\alpha \\
\beta
\end{array}
\right),\quad X \left(
\begin{array}{c}
\bar{\beta} \\
\bar{\alpha}
\end{array}
\right)\label{eq:canonical}
\end{align}
where $\bar{\alpha}$ and $\bar{\beta}$ are duals of $\alpha$ and $\beta$ obtained by taking $x \rightarrow x^{-1}$, $y\rightarrow y^{-1}$, $z\rightarrow z^{-1}$. The Hamiltonian is
\begin{align}
H = - \sum_{ijk} Z \left(\begin{array}{c} x^{i}y^{j}z^{k} \alpha \\ x^{i}y^{j}z^{k} \beta \end{array} \right) - \sum_{ijk} X \left(\begin{array}{c} x^{i}y^{j}z^{k} \bar{\beta} \\ x^{i}y^{j}z^{k} \bar{\alpha} \end{array} \right).
\end{align}
One can verify that terms in the Hamiltonian commute with each other due to the use of dual polynomials. So this is a stabilizer Hamiltonian. 

\begin{figure*}[htb!]
\centering
\includegraphics[width=0.7\linewidth]{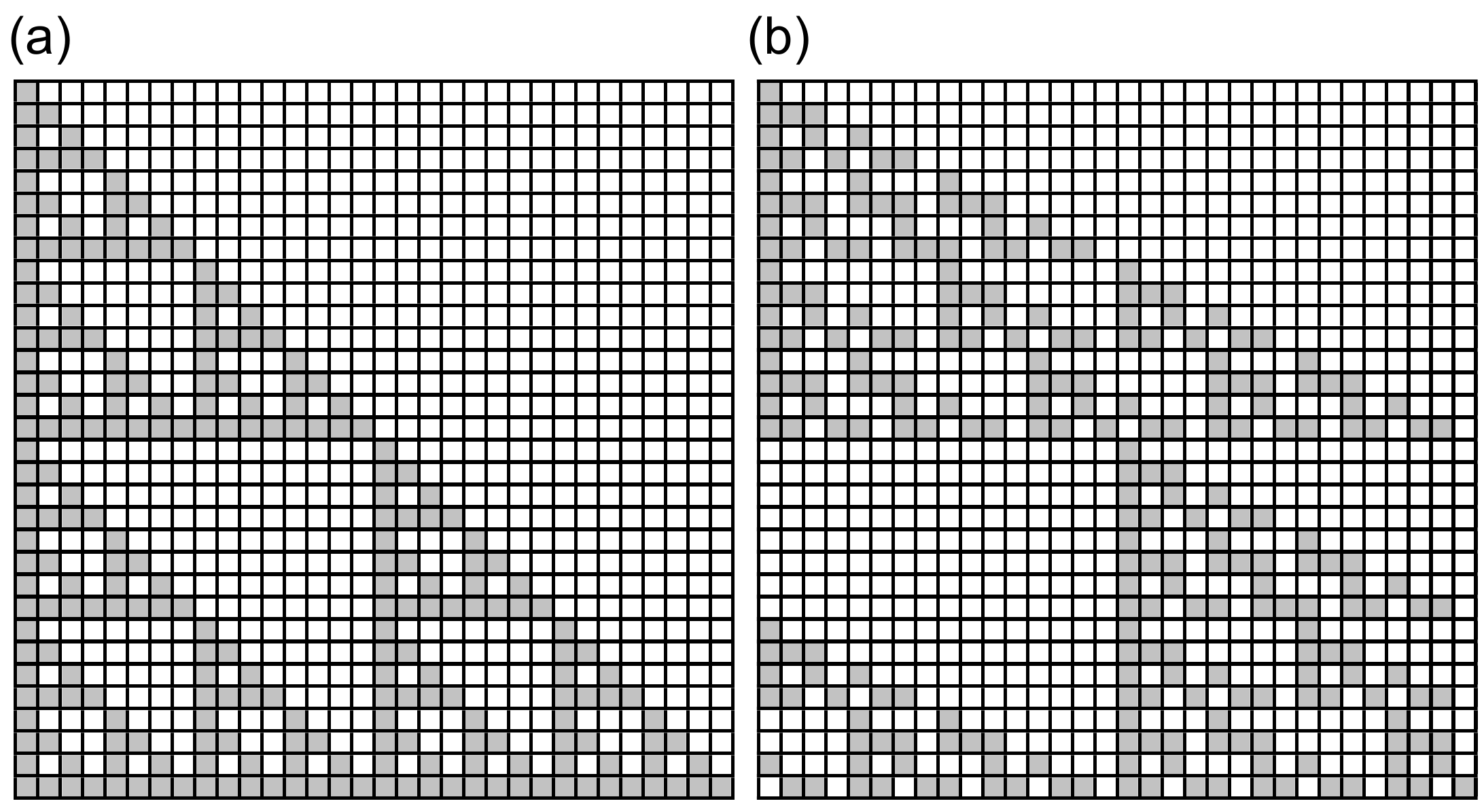}
\caption{(a) The Sierpinski triangle from $f=1+x$. (b) The model from $f=1+x+x^{2}$. 
} 
\label{fig_fractal}
\end{figure*}

In Ref.~\citen{Beni13}, quantum codes with the following polynomials have been studied in detail:
\begin{align}
\alpha = 1 + f(x)y \qquad \beta = 1+ g(x)z. 
\end{align}
Logical operators have fractal shapes which are generated by polynomials $f(x)$ and $g(x)$:
\begin{align}
\tb{f}(x,y) &= 1 + fy + f^{2}y^{2} + \cdots \\ 
\bar{\tb{f}}(x,y) &= 1 + \bar{f}\bar{y} + \bar{f}^{2}\bar{y}^{2} + \cdots\\
\tb{g}(x,z) &= 1 + gz + g^{2}z^{2} + \cdots \\
\bar{\tb{g}}(x,z) &= 1 + \bar{g}\bar{z} + \bar{g}^{2}\bar{z}^{2} + \cdots.
\end{align}
Note $\tb{f}(x,y)$ lives on a $(\hat{x},\hat{y})$-plane while $\tb{g}(x,z)$ lives on a $(\hat{x},\hat{z})$-plane. Here $\tb{f}(x,y)$ and $\tb{g}(x,y)$ have fractal shapes in general as shown in Fig.~\ref{fig_fractal}. Namely, for $f=1+x$, one obtains the well-known Sierpinski triangle while for $f=1+x+x^{2}$, one obtains a fractal shape with Hausdorff dimension $\log{(1+\sqrt{5})}/\log{2}$.

Quantum fractal codes typically have $k=2L$, and there are $2L$ of $Z$-type logical operators and $2L$ of $X$-type logical operators:
\begin{align}
&\ell^{(Z)}_{i}=
Z\left( 
\begin{array}{cc}
0 \\
x^{i}\tb{f}(x,y)
\end{array}
\right)\quad \ r^{(Z)}_{i}=
Z\left( 
\begin{array}{cc}
x^{i}\tb{g}(x,z)  \\
0 
\end{array}
\right)\\
&\ell^{(X)}_{i}=
X\left( 
\begin{array}{cc}
x^{i}\bar{\tb{f}}(x,y) \\
0  
\end{array}
\right)\quad 
r^{(X)}_{i}=
X\left( 
\begin{array}{cc}
0\\
x^{i}\bar{\tb{g}}(x,z)    
\end{array}
\right)
\end{align}
where $i=0,\cdots,L-1$. Therefore, $Z$-type logical operators have geometric shapes of $\tb{f}(x,y)$ and $\tb{g}(x,z)$ while $X$-type logical operators have geometric shapes of $\bar{\tb{f}}(x,y)$ and $\bar{\tb{g}}(x,z)$. For a technical subtlety concerning the number of logical qubits, see Ref.~\citen{Beni13}.

The cubic code corresponds to the choices of
\begin{equation}
\begin{split}
\alpha &= 1 + (1+x+x^{2})y \\ 
\beta &= 1 + (1+x)z + (1+x+x^{2})z^{2}.\label{eq:choice}
\end{split}
\end{equation}
Note that polynomials $\alpha,\beta$ generate fractal geometries. Loosely speaking, if $\alpha$ and $\beta$ generate different fractal geometries, then the system is free from string-like logical operators~\cite{Beni13}. Thus there is no special meaning in choosing $\alpha,\beta$ to be as in Eq.~(\ref{eq:choice}), and our arguments apply to a larger class of translation symmetric stabilizer codes. 

\subsection{Thermal instability}

The fact that fractal logical operators can be supported on two-dimensional planes implies that excitations may propagate according to fractal geometries confined on a plane. We shall concentrate on excitations associated with violations of $X$-type stabilizers. Excitations can be characterized by a polynomial $E$:
\begin{align}
E(x,y,z) = \sum_{i,j,k} e_{i,j,k} x^{i}y^{j}z^{k}
\end{align}
where $e_{i,j,k}=1$ means that an excitation is present at $(i,j,k)$ while $e_{i,j,k}=0$ means that an excitation is absent. 

Consider an excitation at $(0,0,0)$, which can be represented by $E=x^0y^0z^0=1$. If one applies a Pauli operator $Z_{0,0,0}^{(B)}$, then the excitation will propagate to $E = fy$ since $Z_{0,0,0}^{(B)}$ would create an excitation characterized by $E=1+fy$. By applying Pauli $Z$ operators in an appropriate manner, one can move this single excitation at $(0,0,0)$ to 
\begin{align}
E = f^{m}y^{m}
\end{align}
in the $\hat{y}$ direction by forming a fractal geometry. See an example for $f=1+x$ shown in Fig~\ref{fig_1x}. 

Let us assume that a stabilizer term at $(0,0,0)$ is removed. Then, excitations $f^{m}y^{m}$ can be eliminated by bringing them to $(0,0,0)$. This is the key property we will use in our proof. Our strategy is to prove that an excitation $(0,0,0)$ can be eliminated with high probability by some local unitary operation when stabilizer terms are randomly removed. 

\begin{figure}[htb!]
\centering
\includegraphics[width=0.45\linewidth]{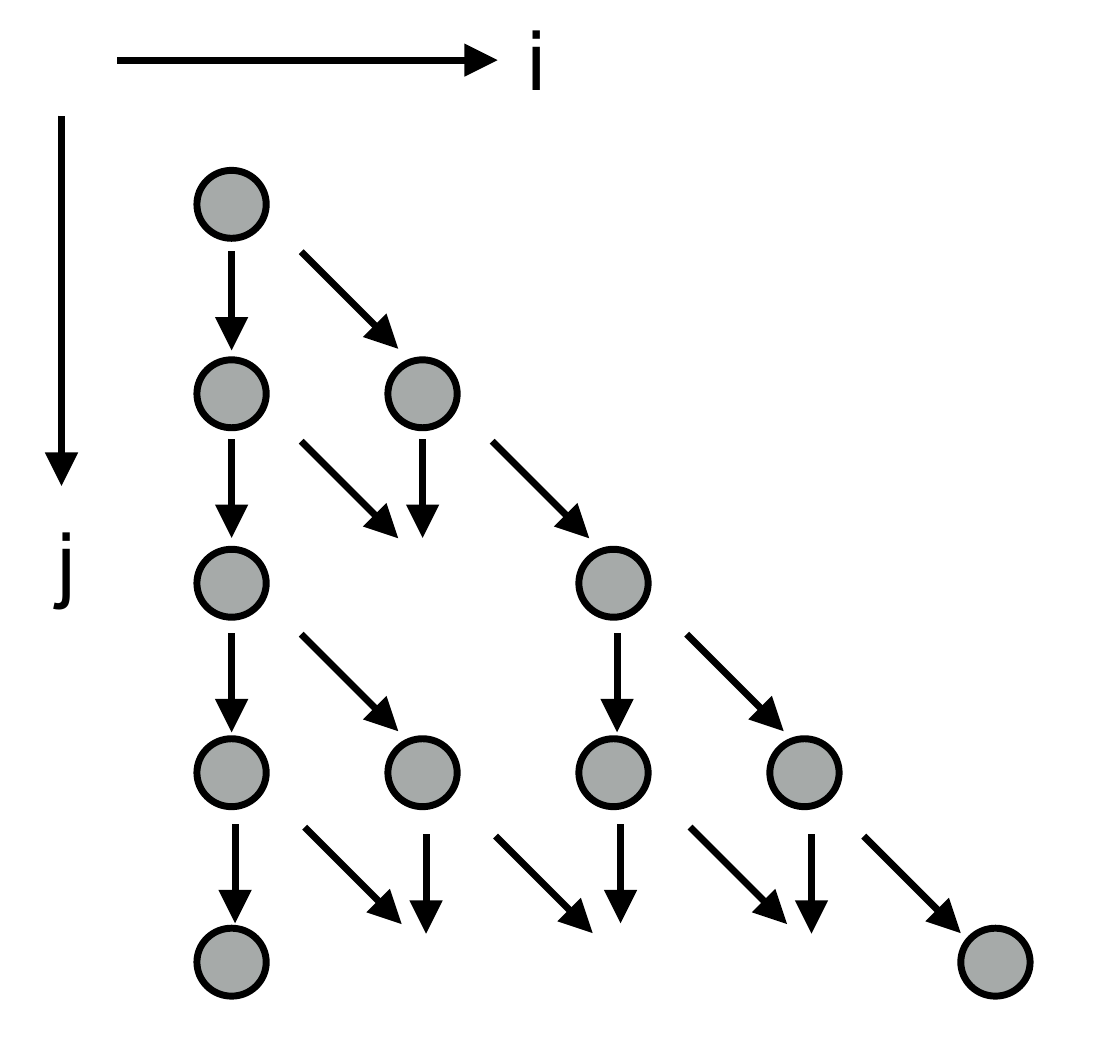}
\caption{Propagations of excitations according to a fractal geometry.
} 
\label{fig_1x}
\end{figure}

We now prove that, for $f(x)=1+x$, the system is not topologically ordered at nonzero temperature. We need to introduce a few notations. Since one can move excitations inside a two-dimensional plane, we concentrate on a plane $(i,j,0)$ by setting $k=0$. We will also write $(i,j,0)$ as $(i,j)$ for brevity of notation. Consider the following sets of sites 
\begin{align}
B =\{ (i,j) :0\leq i,j \leq r-1, \ \ i \leq j   \}.
\end{align}
We want to prove that a single excitation at $(0,0)$ can be almost surely eliminated inside a block $B$ for $r = O(\log(L)^{2})$. It is convenient to define 
\begin{align}
B_{t} &=\{ (i,j)\in B : i \geq r-t \} \\ 
L_{t}&=\{ (i,j)\in B : i = r-t \}. 
\end{align}
See Fig.~\ref{fig_region} for graphical representations of $B,B_{t},L_{t}$. 

\begin{figure}[htb!]
\centering
\includegraphics[width=0.9\linewidth]{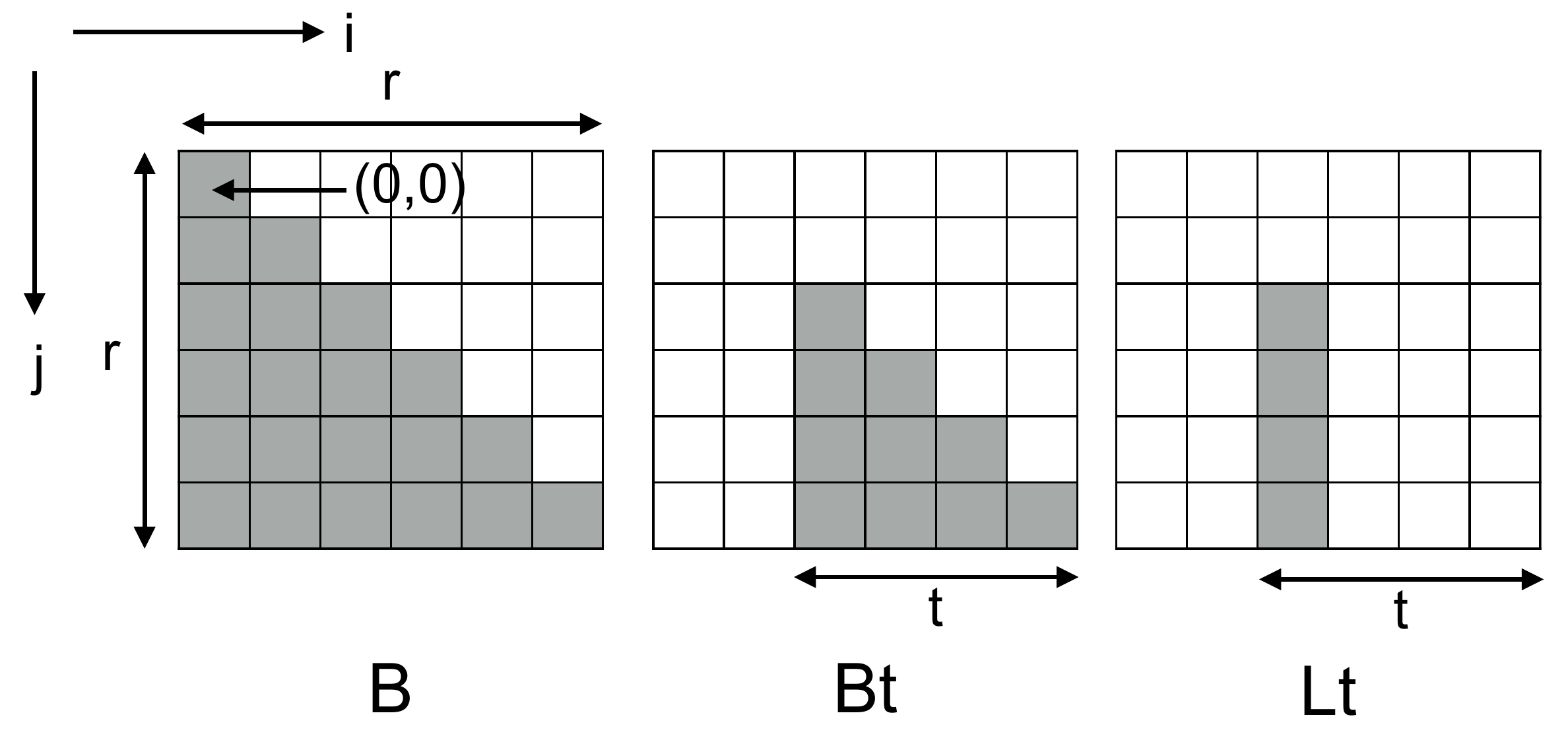}
\caption{Graphical representations of $B,B_{t},L_{t}$. 
} 
\label{fig_region}
\end{figure}

Assume that a stabilizer term at $(i,j)$ is removed. Then excitations $x^{i}f^{r-1-j}y^{r-1}$ can be eliminated since a single excitation at $(i,j)$ would propagate to $x^{i}f^{r-1-j}y^{r-1}$. Note that this excitation is located at the bottom of $B$. Let us consider a subset of $B$, denoted by  $B^{*}$, whose elements are chosen randomly with probability $p>0$. Also, define $L_{t}^{*}= L_{t} \cap B^{*}$ and $B_{t}^{*}= B_{t} \cap B^{*}$. We shall consider the following linear spaces spanned by polynomials
\begin{align}
F_{t}= \left\langle \ x^{i}f^{j} : (i,j)\in B_{t}^{*} \ \right\rangle.
\end{align}
The dimension of $F_{t}$, denoted by $d_{t}$, is defined as the number of independent polynomials in $F_{t}$. Also, define 
\begin{align}
G_{t} = \left\langle \ x^{i} : i \geq r-t \ \right\rangle.
\end{align}
Then $F_{t}\subseteq G_{t}$, and $d_{t}\leq t$. If $d_{r}=r$, then any excitation inside $B$ can be eliminated. 


\begin{lemma}\label{lemma_cubic}
For any $p > 0$ there exists a sufficiently large but finite constant $c$ such that, for $r=\frac{c}{p^3}\cdot \log(L)$, 
\begin{align}
1 - \Prob[d_{r}= r] \leq O\left(\frac{1}{L^{c'}}\right)\label{eq:prob}
\end{align}
for some constant $c' > 3 $ which is independent of $L,p$ and can be made arbitrarily large by taking large $c$.
\end{lemma}

This lemma implies that a single excitation inside $B$ can be eliminated with high probability. Note that there are $O(L^{3})$ terms in the cubic code while Eq.~(\ref{eq:prob}) implies that the probability approaches to unity faster than $O(1/L^{3})$. Thus, an imperfect Hamiltonian $H(p)$ is almost surely range-$\log(L)$ trivial, proving the thermal instability of the original Hamiltonian $H$. 

\begin{proof}
Consider an arbitrary set of sites $\hat{B}_{t}\subseteq B_{t}$ and denote the dimension of a linear space spanned by polynomials generated from $\hat{B}_{t}$ by $\hat{d}_{t}$. Let us choose sites in $L_{t+1}^{*}\subseteq L_{t+1}$ randomly with probability $p$. We define $\hat{B}_{t+1}=L_{t+1}^{*} \cup \hat{B}_{t}$ and denote the dimension of the linear space by $\hat{d}_{t+1}$. From the definition, one has $\hat{d}_{t+1}\geq \hat{d}_{t}$. Also note that $d_{t}\leq t$. We are interested in the probability distribution of the increase: $\hat{d}_{t+1}- \hat{d}_{t}$ for a given $\hat{B}_{t}$.

First, assume that $\hat{d}_{t}$ is already at its maximal value: $\hat{d}_{t}=t$. If $L^{*}_{t+1}$ is an empty set, then $\hat{d}_{t+1}=\hat{d}_{t}$. Since $x^{i}f^{j}\not\in F_{t+1}$ for all $(i,j)\in L_{t+1}$, if $L^{*}_{t+1}$ is a non-empty set, one has $\hat{d}_{t+1}= \hat{d}_{t}+1$. So one has
\begin{equation}
\begin{split}
&\mbox{Prob}[\hat{d}_{t+1}=t] = (1-p)^{t+1} \\
&\mbox{Prob}[\hat{d}_{t+1}= t+1] = 1- (1-p)^{t+1}.
\end{split}
\end{equation}
Note that the above expressions are valid for any sets $\hat{B}_{t}$ with $\hat{d}_{t}=t$. Thus, for sufficiently large $t$, the value of $d_{t}$ will almost surely increase by one. This implies that once $d_{t}$ reaches its maximal value at some $t=t_{0}$, then we have $d_{t}=t$ for all $t>t_{0}$ with very high probability. Namely, if $t=\frac{c}{p}\log(L)$ for sufficiently large $c>0$, then $\mbox{Prob}[\hat{d}_{t+1}=t]$ is polynomially suppressed. Later, we will take $r=\frac{c'}{p^3}\cdot \log(L)$, so $\mbox{Prob}[\hat{d}_{t+1}=t]$ will be negligible for all $t\geq \frac{c}{p}\log(L)$.

Next, we consider sets $\hat{B}_{t}$ with $\hat{d}_t<t$. We find a lower bound on the probability that $\hat{d}_t$ increases by more than one. Let $e_{j}=x^{r-t}f^{r-1-j}$. Note that $e_{j}$ are independent for all $j=0,\ldots,t$. Also note that $e_{j}\not \in F_{t}$ since $e_{j}$ always contains $x^{r-t}$ while polynomials in $F_{t}$ do not contain it. For a given $i$ ($0\leq i \leq t$), there always exists $k(i)$ ($0\leq k(i) \leq t$, $i\not=k(i)$) such that $e_{i} + e_{k(i)} \not \in F_{t}$ since $e_{j}$ are independent. To prove this statement, let us suppose that $e_{i}+e_{j}\in F_{t}$ for all $j\not=i$. Since $\hat{d}_{t}<t$ and $\langle e_{0},\ldots,e_{t} \rangle=G_{t+1}$, there must exist a pair of polynomials $a,b\in G_{t+1}$ such that 
\begin{align}
a,b,a+b \not\in F_{t}. \label{eq:elements}
\end{align}
Note that $a,b$ can be expanded by $e_{0},\ldots,e_{t}$. Since adding/subtracting elements from $F_{t}$ to $a,b$ does not affect Eq.~(\ref{eq:elements}), we can choose $a,b$ to be $a=e_{i}$ and $b=e_{k(i)}$ where $k(i)\not= i$ is some integer. But one has $a+b=e_{i}+e_{k(i)}\in F_{t}$, leading to a contradiction. So, if $(r-t,i),(r-t,k(i))\in L^{*}_{t+1}$, then $\hat{d}_{t+1}\geq \hat{d}_{t}+2$. Such a probability is lower bounded by $p^{2}$:
\begin{equation}
\Prob[\hat{d}_{t+1} \geq \hat{d}_t + 2] \geq p^2.
\end{equation}
This lower bound is valid for any sets $\hat{B}_{t}$ with $\hat{d}_{t}<t$. Also, recall that $\Prob[\hat{d}_{t+1} =\hat{d}_t ]= (1-p)^{t+1}$, which is very small for large $t$. Therefore, $\hat{d}_{t}$ increases by at least one almost surely and increases by more than one with probability at least $p^2$ (see Fig.~\ref{fig_transition}). 

\begin{figure}[htb!]
\centering
\includegraphics[scale=0.33]{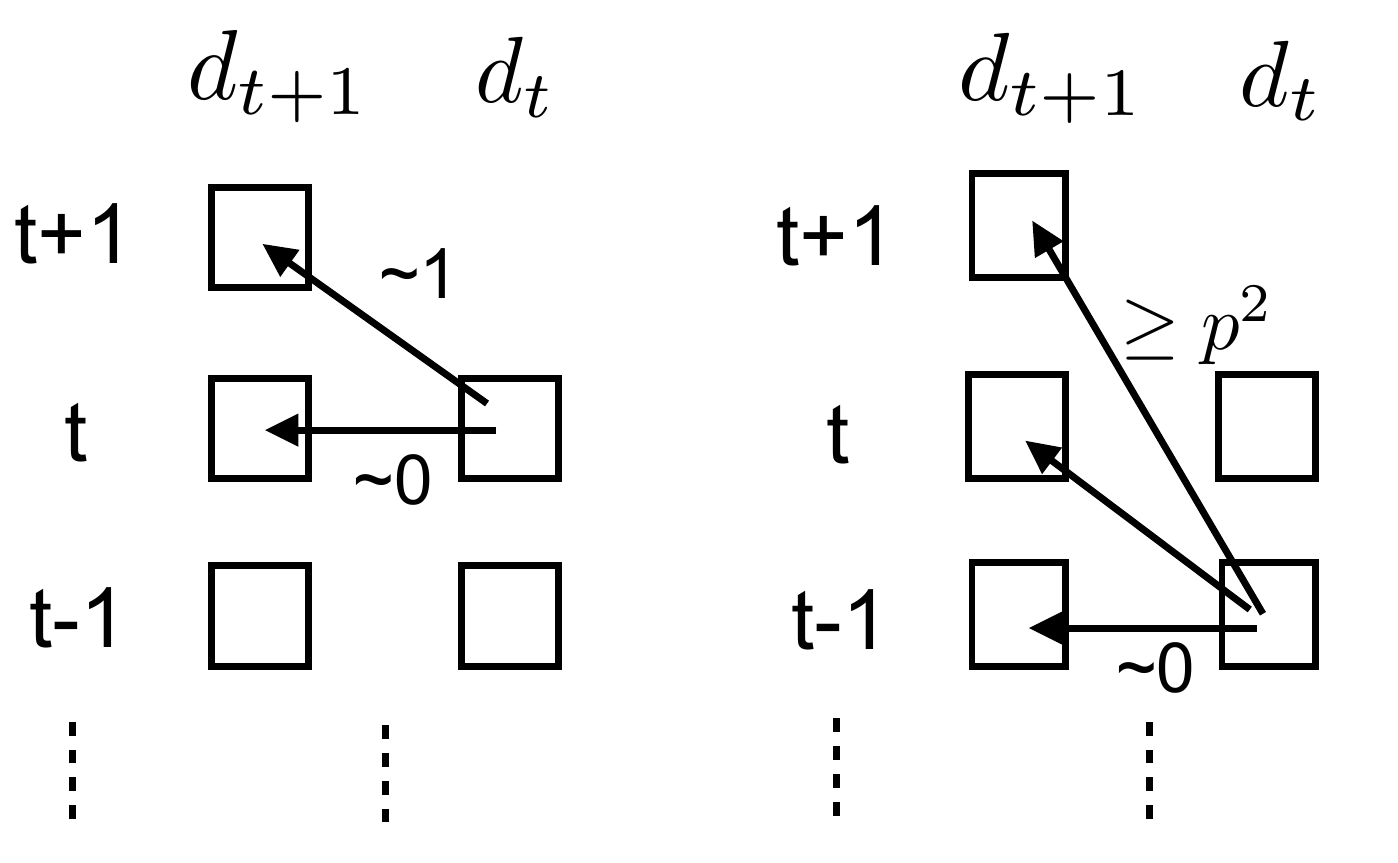}
\caption{Evolutions of probability distributions of $d_{t}$.
} 
\label{fig_transition}
\end{figure}

Intuitively, since the probability of $\hat{d}_{t}$ increasing by more than 1 is constant while the probability of $\hat{d}_{t}$ failing to increase is exponentially vanishing, a difference of $t_0 - \hat{d}_{t_0}$ (for $\hat{d}_{t_0} < t_0$) can be made up for in a large number of subsequent steps with low probability of introducing too many additional failures. More formally, let us consider an arbitrary set $\hat{B}_{t_{0}}$ for $t_{0}=\frac{c}{p}\log(L)$ with sufficiently large but finite $c>0$ so that $\hat{d}_{t}$ will almost surely increase by at least one for all $t\geq t_{0}$. In order for $d_{r}=r$, $d_{t}$ needs to reach its maximal value for some $t$ with $t_{0}\leq t \leq r$. For each $t$, the probability of increase by more than one is lower bounded by $p^{2}$. So the probability of $d_{r}$ not reaching its maximal value is upper bounded by 
\begin{align}
\sum_{j=0}^{t_{0}-1} \left(
\begin{array}{c}
r- t_{0} \\
j
\end{array}
 \right) (p^{2})^j (1-p^2)^{r-t_{0}-j}
\end{align}
which corresponds to a lower tail of the binomial distribution. For $r=\frac{c'}{p^3}\log(L)$ with sufficiently large $c'>0$, this expression is polynomially suppressed with respect to $L$. Thus, $d_{r}$ reaches its maximal value $r$ with probability approaching to unity, and all the errors are polynomially suppressed with respect to $L$. This completes the proof. 
\end{proof}

This implies that there almost surely exists a unitary quantum circuit acting solely inside $B$ that can transform all stabilizer operators so that they are diagonal in the same basis ($X$ or $Z$). This probability goes to $1$ in the thermodynamic limit ($L\rightarrow\infty$). Thus all the terms in $H(p)$ are almost surely $O(\log(L))$-removable. We have shown that quantum fractal codes with $f=1+x$ are not topologically ordered at nonzero temperature. The argument does not rely on the fact that the system consists of qubits, and thus can be generalized to quantum fractal codes with qudits. 

Finally, we sketch the proof for generic $f(x)$. Without loss of generality, one can assume that $f(x)$ can be written as $f(x)=\sum_{i=0}^{m}c_{i}x^{i}$ and $c_{0},c_{m}\not=0$ where $m$ is a positive integer. We consider a step-like region $B$, and then define $B^{*},B_{t},L_{t},B^{*}_{t},L^{*}_{t}$ in a similar manner as depicted in Fig.~\ref{fig_region_2} where each step consists of $m$ columns. We then define linear spaces by
\begin{align}
F_{t} = \langle x^{i} f^{j} : (i,j)\in B_{t}^{*} \rangle
\end{align}
and denote its dimension by $d_{t}$. Note that $d_{t}\leq mt$. We want to prove that $d_{r}=mr$ with high probability. 

When $t$ is logarithmically large with respect to $L$, we will have $d_{t+1}-d_{t}< m$ with polynomially vanishing probability. So, at each step, $d_{t}$ will almost surely increase by at least $m$. Also, one finds that the probability of having $d_{t+1}-d_{t} > m$ is lower bounded by $p^{m+1}$ as long as $d_{t}$ has not reached its maximal value by using an argument similar to the proof of lemma~\ref{lemma_cubic}. Note that the cubic code has $f=1+x+x^{2}$ with $m=2$, so we should take $r=\frac{c}{p^3}\log(L)$ for some large but finite $c>0$. Then the probability of $d_{r}$ being at its maximal value approaches to unity polynomially. Thus, the cubic code is not topologically ordered at nonzero temperature. This proof strategy can be generalized to cases where $\alpha,\beta$ have the form of $1+f_{1}(x)y+f_{2}(x)y^2+\cdots$ and $1+g_{1}(x)z+g_{2}(x)z^2+\cdots$ by appropriately modifying the constructions of $B^{*},B_{t},L_{t},B^{*}_{t},L^{*}_{t}$. However, at this moment, we do not know how to generalize this proof to arbitrary $\alpha,\beta$. It would be interesting to generalize our studies to arbitrary stabilizer codes with translation symmetries. Also it should be noted that approximating the Gibbs ensemble of the toric code requires $O(\sqrt{\log(L)})$-range circuit only while approximating the Gibbs ensemble of  the cubic code requires $O(\log(L))$-range circuit in our approach.

\begin{figure}[htb!]
\centering
\includegraphics[scale=0.30]{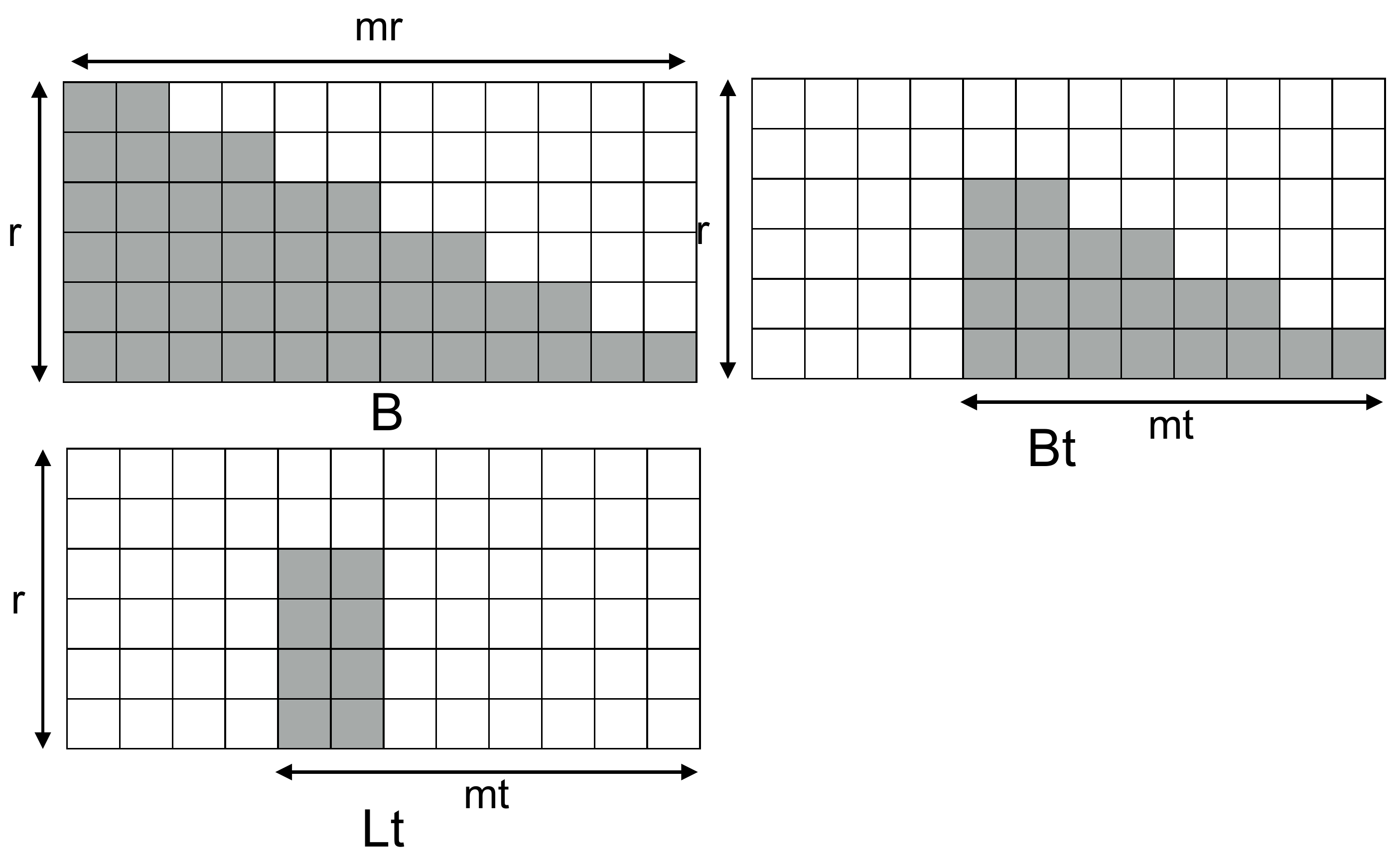}
\caption{Graphical representations of $B,B_{t},L_{t}$ for $m=2$. 
} 
\label{fig_region_2}
\end{figure}

\section{Memory time}\label{sec:decomposition}

Having studied static properties of quantum memories, we now wish to switch gears to study dynamical properties of quantum memories by computing upper bounds on memory time, the expected lifetime of quantum information. In this section, we begin by reviewing a mathematical framework to treat thermalization of stabilizer quantum memory Hamiltonians based on Davies generators. We then present a method of choosing minimal coupling to the thermal bath, minimal in the sense that the system-bath interaction is as simple as possible but still yields an ergodic reduced dynamics. 



\subsection{Lindblad operators}

We shall model the dynamics of the system-bath interaction using the Davies weak-coupling limit~\cite{Davies74} and begin by reviewing this construction. The total Hamiltonian describing the system and the bath is given by
\begin{equation}
H = H^\text{sys} + H^\text{bath} + H^\text{int}, \quad H^\text{int}  = \sum_\alpha S_\alpha \otimes f_\alpha
\end{equation}
where $S_\alpha$ are system operators, such as single-qubit Pauli operators, and $f_\alpha$ are bath operators. The density matrix of the system obeys
\begin{equation}
\dot{\rho}(t) = -i[H^\text{sys},\rho(t)] + \mathcal{L}(\rho(t))
\end{equation}
with $\rho(0)$ the initialized ground state. 

The Lindblad generator describing transfer of energy between the system and the bath is given by
\begin{multline*}
\mathcal{L}(\rho(t)) = \sum_{\alpha,\omega}h(\alpha,\omega)\big(S_{\alpha}(\omega)\rho(t)S_{\alpha}(\omega)^\dagger  \\
 -\frac{1}{2}\{\rho, S_{\alpha}(\omega)^\dagger S_{\alpha}(\omega)\}\big)
\end{multline*}
with $S_\alpha(\omega)$ describing transitions from eigenvectors of $H^\text{sys}$ with energy $E$ to eigenvectors with energy $E - \omega$ with $\omega$ delivered to the bath and $h(\alpha,\omega)$ describing the rate of these transfers. To satisfy detailed balance, we demand
\begin{equation}
h(\alpha,-\omega) = e^{-\beta\omega}h(\alpha,\omega) \label{eq:detbal}
\end{equation}
with $\beta = 1/T$, which ensures that the thermal density matrix $\rho_\beta \sim e^{-\beta H}$ is a fixed point of the dynamics ($\mathcal{L}(\rho_\beta) = 0$). 

Moreover, if the action of $\mathcal{L}$ is ergodic, meaning that all the eigenvectors can be visited, then $\rho_{\beta}$ will be the \textit{only} fixed point. So if a minimal interaction Hamiltonian is chosen, it is necessary to ensure that it is ergodic. The condition for this is given in Ref.~\citen{Spohn77,Frigerio78} by 
\begin{equation}
\{S_\alpha, H^\text{sys}\}' = \mathbb{C}1. \label{eq:ergodicity}
\end{equation}
In other words, no system operator apart from multiples of the identity should commute with all of the $S_\alpha$ and the system Hamiltonian so that all energy levels of the system Hamiltonian are accessible. For a stabilizer Hamiltonian, let us assume that $S_{\alpha}$ are Pauli operators. Then the condition for ergodicity is given by
\begin{align}
\left\langle \{S_{\alpha}\}, H^\text{sys} \right\rangle = \mathcal{P}.
\end{align}
In other words, $S_{\alpha}$ and stabilizer generators in $H^\text{sys}$ need to form the full Pauli operator group since otherwise there exists a non-trivial Pauli operator which commutes with $S_{\alpha}$ and $H^\text{sys}$. A natural choice of $\{S_{\alpha}\}$ is to pick Pauli $X$ and $Z$ operators acting on each qubit. Then the Davies generator has the form 
\begin{align}
\mathcal{L} = \sum_{j=1}^{n}\mathcal{L}_{X,j} + \sum_{j=1}^{n}\mathcal{L}_{Z,j} = \mathcal{L}_{X} + \mathcal{L}_{Z}.\label{eq:CSS_generator}
\end{align}

The convergence time to the Gibbs ensemble is given by the inverse of the gap of the superoperator $\mathcal{L}$. When the Davies generator has the form of Eq.~(\ref{eq:CSS_generator}) and the Hamiltonian is described by a CSS stabilizer code, convergence time is given by $\tau^{-1}_{\text{conv}}=\min( \text{Gap}(\mathcal{L}_{X}), \text{Gap}(\mathcal{L}_{Z}) )$. Our goal is to obtain lower bounds on quantum memory time of CSS codes. Since memory systems suffer from both $X$ and $Z$ errors, memory time is given by $\tau^{-1}_{\text{memory}}=\max( \text{Gap}(\mathcal{L}_{X}), \text{Gap}(\mathcal{L}_{Z}) )$. Note that it is possible that memory time and convergence time are different. For instance, in the three-dimensional toric code, convergence time is $O(\exp(L))$ while quantum memory time is $O(1)$. See~\cite{Alicki09} for details.

To obtain an upper bound on quantum memory time, or equivalently lower bounds on gaps in $\mathcal{L}_{X}$ and $\mathcal{L}_{Z}$, we pick a smaller, but ergodic $\{S^{\text{min}}_{\alpha}\}$. Observe that one can remove some of  Pauli operators from $\{S_{\alpha}\}$ while retaining its ergodicity. The gap for $\{S_{\alpha}\}$ will be lower bounded by the one for $\{S^{\text{min}}_{\alpha}\}$ since choosing a smaller set always makes the gap smaller. In Ref.~\citen{Alicki09}, Alicki \emph{et. al.} used this strategy to provide an upper bound on quantum memory time of the two-dimensional toric code, giving a rigorous proof that the model is not self-correcting. Namely, they choose $\{S^{\text{min}}_{\alpha}\}$ carefully such that the problem is reduced to computing classical memory time of one-dimensional Ising model.

\subsection{Ergodic Decomposition}

In this section, we generalize an ergodic decomposition of the Davies generator applied by Alicki, \emph{et. al.} to the toric code in Ref.~\citen{Alicki09} to all CSS codes. 

\begin{lemma}
Given a CSS stabilizer code defined by stabilizer group $\mathcal{S}$ with $k$ logical qubits, there exists a tripartition of the full system into three non-overlapping sets of qubits, denoted by $L$, $R_X$, $R_Z$ such that
\begin{equation}
\begin{split}
&\mathcal{P}^{(X)} = \langle \mathcal{S}_{X}, \mathcal{P}^{(X)}(L\cup R_X)\rangle \\
&\mathcal{P}^{(Z)} = \langle \mathcal{S}_{Z}, \mathcal{P}^{(Z)}(L\cup R_Z)\rangle 
\end{split}
\end{equation}
where $\mathcal{P}^{(X,Z)}$ is a group of Pauli-$X,Z$ operators, $\mathcal{P}^{(X,Z)}(R)$ is a group of Pauli-$X,Z$ operators acting on region $R$, and 
$$|L| = k \quad |R_X| = \log_2 |\mathcal{S}_{Z}| \quad |R_Z| = \log_2 |\mathcal{S}_{X}|.$$
\end{lemma}

Thus choosing 
\begin{equation}
\{S_\alpha\} = \mathcal{P}^{(X)}(R_X\cup L) \cup \mathcal{P}^{(Z)}(R_Z\cup L)
\end{equation}
for the noise model's Lindblad generator satisfies the ergodicity requirement in Eq.~(\ref{eq:ergodicity}). Note that the choice of $\{S_{\alpha}^{\text{min}}\}$ is not unique, but the lemma teaches us the minimal number of generators to ensure the ergodicity. The proof also provides an efficient algorithm (running in polynomial time) to find $\{S_{\alpha}^{\text{min}}\}$, which is similar to the one presented in Ref.~\citen{Beni10}.

\begin{proof}
Denote independent generators of $\mathcal{S}_{Z}$ as $S^{(Z)}_j$. We seek a set of operators $\{A^{(X)}_j\}$ such that $S^{(Z)}_i$ and $A^{(X)}_j$ anticommute if $i=j$ and commute otherwise. We shall therefore proceed by induction on $j$. For $S^{(Z)}_1$, we can find a single qubit Pauli-$X$ operator $A^{(X)}_1$ such that $\{S^{(Z)}_1,A^{(X)}_1\}=0$. However, $A^{(X)}_1$ may anticommute with other generators $S^{(Z)}_j$ for $j > 1$, so for all $j > 1$, if $\{S^{(Z)}_j,A^{(X)}_1\}=0$, we update $S^{(Z)}_j\rightarrow S^{(Z)}_1 S^{(Z)}_j$. This resolves all of the commutation relations for $A^{(X)}_1$, thereby completing the base case. The commutation relations so far can be concisely summarized in the following canonical form:
\begin{equation}
\left\langle\begin{array}{cccc}
S^{(Z)}_1 & S^{(Z)}_2 & \dots & S^{(Z)}_{n_Z} \\
A^{(X)}_1 & & &  \\
\end{array}\right\rangle
\end{equation}
with $n_Z = |R_X| = \log_2 |\mathcal{S}_Z|$. Assume that there exist $m$ operators $A^{(X)}_j$ with $1\leq j \leq m$ satisfying the desired commutation relations. 
\begin{equation}
\left\langle\begin{array}{cccccc}
S^{(Z)}_1 & \dots & S^{(Z)}_m & S^{(Z)}_{m+1} & \dots & S^{(Z)}_{n_Z} \\
A^{(X)}_1 & \dots & A^{(X)}_m  & & &  \\
\end{array}\right\rangle
\end{equation}
The group $\langle A^{(X)}_1  \dots  A^{(X)}_m \rangle$ is a set of Pauli-$X$ operators acting on $m$ qubits. For $S^{(Z)}_{m+1}$, pick a single-qubit Pauli-$X$ operator $A^{(X)}_{m+1}$ such that $\{S^{(Z)}_{m+1}, A^{(X)}_{m+1}\} = 0$. For $j\leq m$, if $\{S^{(Z)}_{j}, A^{(X)}_{m+1}\} = 0$ we update $A^{(X)}_{m+1}\rightarrow A^{(X)}_{j}A^{(X)}_{m+1}$. For $j > m+1$ if $\{S^{(Z)}_{j}, A^{(X)}_{m+1}\} = 0$, we update $S^{(Z)}_{j}\rightarrow S^{(Z)}_{m+1}S^{(Z)}_{j}$. This ensures that the desired commutation relations are satisfied for $A^{(X)}_j$ for $1\leq j \leq m+1$:
$$\left\langle\begin{array}{cccccc}
S^{(Z)}_1 & \dots & S^{(Z)}_m & S^{(Z)}_{m+1} & \dots & S^{(Z)}_{n_Z} \\
A^{(X)}_1 & \dots & A^{(X)}_m & A^{(X)}_{m+1} & &  \\
\end{array}\right\rangle$$
Further, the commutation relations and choice of operators for $A^{(X)}_j$ guarantees that $\langle A^{(X)}_1  \dots  A^{(X)}_{m+1} \rangle$ must be the group of Pauli-$X$ operators acting on $m+1$ qubits. This completes the inductive step to give, in all,
\begin{equation}
\left\langle\begin{array}{ccc}
S^{(Z)}_1 & \dots & S^{(Z)}_{n_Z} \\
A^{(X)}_1 & \dots & A^{(X)}_{n_Z}  \\
\end{array}\right\rangle
\end{equation}
Define the group $\mathcal{P}^{(X)}(R_X)$ as 
\begin{equation}
\mathcal{P}^{(X)}(R_X) = \langle A^{(X)}_1  \dots A^{(X)}_{n_Z} \rangle
\end{equation}
where $R_X$ is a set of $n_Z$ qubits. Suppose that we now expand $\mathcal{S}_Z$ by adding in $k$ independent $Z$-type logical operators $\bar{Z}_i$:
$$\left\langle\begin{array}{cccccc}
S^{(Z)}_1 & \dots & S^{(Z)}_{n_Z} & \bar{Z}_1 & \dots & \bar{Z}_k\\
A^{(X)}_1 & \dots & A^{(X)}_{n_Z} & & & \\
\end{array}\right\rangle$$
Then, via the algorithm described above, we can find $k$ more operators $A^{(X)}_j$ defined on another region $L$ disjoint from $R_X$ with $|L| = k$. The full Pauli-$X$ group $\mathcal{P}^{(X)}$ is thus given by
\begin{equation}
\mathcal{P}^{(X)} = \left\langle \mathcal{S}^{(X)}, \mathcal{P}^{(X)}(R_X \cup L) \right\rangle.
\end{equation}
Indeed, one can confirm this by observing that $k+n_z+n_x=n$, and $S_{j}^{(Z)}\not\in \mathcal{P}^{(X)}(R_X \cup L)$ since $S_{j}^{(Z)}$ commutes with $S_{j}^{(Z)}$ and $\bar{Z}_{j}$.

We have shown that $\langle A^{(X)}_1  \dots  A^{(X)}_{n_{z}} \rangle$ is a set of Pauli-$X$ operators acting on $n_{z}$ qubits. Let us denote Pauli-$X$ and Pauli-$Z$ operators acting on the $j$-th qubit in $R_{Z}$ by $X_j$ and $Z_{j}$ ($j=1,\ldots,n_{z}$). Note that stabilizer generators $S_{j}^{(Z)}$ are independent and anti-commute with at least one of Pauli-$X$ operators $X_{i}$ in $R_{Z}$. We can thus relabel stabilizer generators by taking their products such that $S^{(Z)}_j = Z_jU^{(Z)}_j$ where $Z_j$ is Pauli-$Z$ acting on the $j$-th qubit and $U^{(Z)}_j$ is some product of Pauli-$Z$ operators acting somewhere on $\overline{R_X}$. Therefore, if we want to generate $\mathcal{P}^{(Z)}$ using $\mathcal{S}_Z$, we only need $\mathcal{P}^{(Z)}(\overline{R_X})$, noting that $L\subseteq \overline{R_X}$. The lemma is thus proved.
\end{proof}

As an example of the decomposition, let us consider the two-dimensional toric code by following Alicki \emph{et al}. Define a pair of comb-like regions $R_{X},R_{Z}$ as depicted in Fig.~\ref{fig_decomposition}. Also, define a set of two qubits at the corners as $L$. Then one can verify that a decomposition into $R_{X},R_{Z},L$ satisfies the ergodicity condition. One can thus choose $\{S_{\alpha}\}$ so that Pauli $X,Z$ operators act on $R_{X}\cup L, R_{Z}\cup L$ respectively. This reduces the problem to the mixing analysis of a classical Ising model defined on a comb-like region $R_{X}\cup L$. 

\begin{figure}[htb!]
\centering
\includegraphics[scale=0.25]{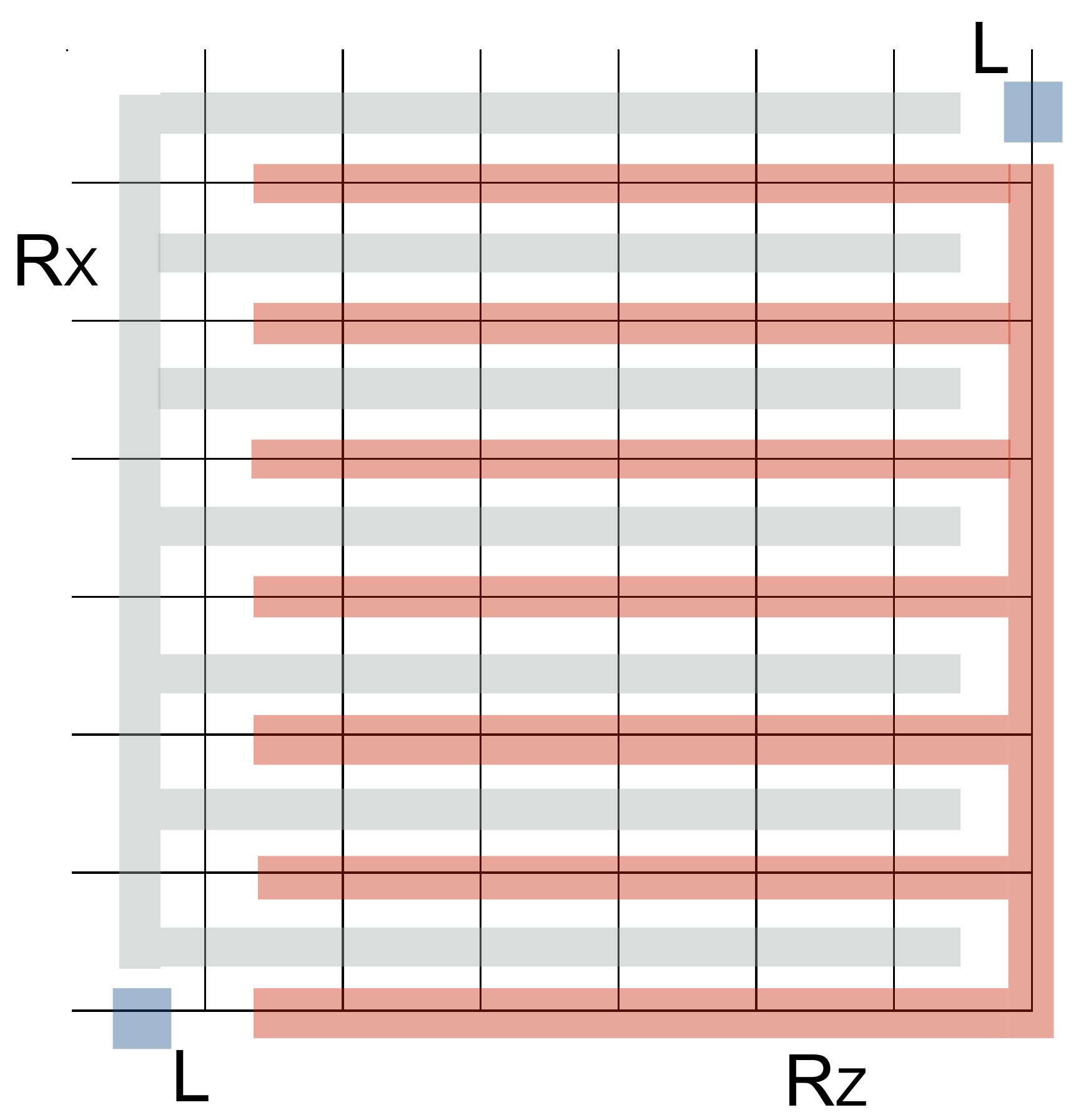}
\caption{Ergodic decomposition for the toric code. Qubits live on edges of the lattice.
} 
\label{fig_decomposition}
\end{figure}

\section{Upper Bound on Memory Time}\label{sec:weld}

\begin{figure}[htb!]
\centering
\includegraphics[scale=0.55]{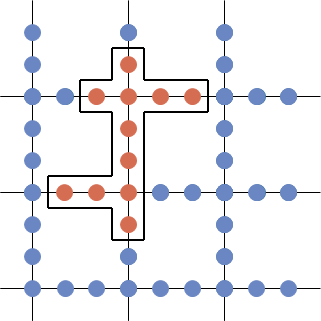}
\caption{Two-dimensional sparse lattice. The lattice has length $L$ and each small square has dimensions $\sqrt{L}\times \sqrt{L}$. A Peierls droplet of frustration is shown in red, marking, for example, a cluster of spin down sites in a sea of spin up sites (blue). This contour has length $\ell=6$, or energy cost $12J$.}
\label{fig_sparse-lattice}
\end{figure}

In this section, we derive an upper bound on quantum memory time of the welded code.

\subsection{The welded code}

The three-dimensional toric code has pairs of one-dimensional $Z$-type logical operators and two-dimensional $X$-type logical operators which characterize propagations of point-like electric charges and loop-like magnetic fluxes. Since the energy barrier for propagations of magnetic loops is $O(L)$, the system is thermally protected from $X$-type errors. Yet, electric charges can propagate freely, so the system is not thermally stable against $Z$-type errors, and thus is not self-correcting. The procedure of welding allows one to introduce energy barrier to propagations of electric charges. Namely, by attaching multiple copies of three-dimensional toric code to form a larger stabilizer code, electric charges need to split in a ``welded'' surface, which leads to a polynomially diverging energy barrier. See appendix~\ref{sec:boundary} for physical interpretation of the welding procedure.

By applying the aforementioned algorithm of finding an ergodic decomposition, we can show that it suffices to study classical memory time of the Ising model supported on a sparse graph, which is similar to the one shown in Fig.~\ref{fig_sparse-lattice}. By a sparse graph, we mean a two-dimensional skeleton-like graph where spins live on the vertices, and vertices with degree 4 are separated by $L^{1-\alpha} - 2$ vertices with degree 2. In Fig.~\ref{fig_sparse-lattice}, we show a two-dimensional example with $\alpha=1/2$, though the analytical derivation below does not depend on the choice of $\alpha$. A physical intuition for this reduction can be obtained by looking at trajectories of electric charges. At each welding surface, point-like electric charges split into multiple charges, and their trajectories form a two-dimensional sparse graph. 

We now bound the memory time of the welded code by bounding the memory time of a corresponding classical code, conceived by the Pauli-$Z$ terms of the reduced $\mathcal{L}$, which can be found as per the algorithm described above. We henceforth show that the memory time of the welded solid code is maximized at each temperature $T$ for a size $L_\text{max}(T)\propto \exp(c\beta)$ for $c > 0$ and $\beta = T^{-1}$. Further, along the curve described by $L_\text{max}(T)$, this bound follows $\tau_\text{max} \propto \exp(\kappa ' (\exp(\kappa \beta))$ for positive constants $\kappa$ and $\kappa '$.

The reduced Lindblad generator responsible for errors on $X$-type stabilizers ($S_\alpha \in \mathcal{P}^{(Z)}(R_Z\cup L)$) is geometrically equivalent to a sparse lattice and mathematically equivalent to a system of nearest-neighbor, ferromagnetically coupled bits as in Fig.~\ref{fig_sparse-lattice}. For this thermal noise, one could equivalently write the Hamiltonian as 
\begin{equation}
H = -J\sum_{\braket{ij}} Z_i \otimes Z_j \qquad J > 0
\end{equation}
where thermal noises are given by Pauli-$X$ operators.

We shall develop the calculation of classical memory time via the Peierls argument as follows. Let us start with one of the ground states, say $|00\ldots0\rangle$, and couple the system to the thermal bath which would flip spins and create regions with $|1\rangle$ states. We can view these clusters of $|1\rangle$ states as excitation droplets whose contours cost energy as spins are not aligned. If the droplet grow too large or too numerous, then erroneous spin flips may proliferate, leading to the loss of encoded information. The probability of observing each excitation droplet of excitation energy $E$ is roughly given by $e^{-\beta E}$. There are many different shapes of excitation droplets with the same energy $E$, and the probability of observing an excitation droplet of energy $E$ is given by $\Omega\exp(-\beta E)= \exp(-\beta F)$ where $F=E-TS$ with $S=\log \Omega$. Thus, the probability of observing an erroneous spin flip can be estimated by computing free energies associated with excitation droplets.

We pause briefly to interpret this probability in terms of quantum circuit complexity and error-correction. Excitation droplets represent errors which could be eliminated by active error-correction. However, performing error-correction instantaneously would require nonlocal feedback circuits to compute the paths of all quasiparticles in the system. Instead, as in Sec.~\ref{sec:topo}, we divide the system into local grids and are interested in the probability that there exists some circuit acting inside each grid to remove excitations without inducing a global error. The cost of employing local error-correction, such as majority voting in finite neighborhoods, is a nonzero probability of failure. In particular, if errors grow too numerous, the correction circuits, decided with imperfect knowledge of the trajectories of the quasiparticles, may induce a global error on the system. The aforementioned probability in terms of free energy, then, can be viewed as the probability of such a failure. Generally, at high temperature, excitations proliferate and failure is likely, but at low temperature, these configurations are suppressed, so failure is unlikely. The lifetime of information stored initially is then limited by the time taken for this probability to rise. See Sec.~\ref{sec:thermaltopo} for more applications of this connection.

We begin calculating the free energy barrier by defining a Peierls contour on the dual lattice (Fig.~\ref{fig_sparse-lattice}). A contour is defined as a closed path on the dual lattice along which there is a domain wall (one side is spin up, the other side is spin down), and let $\ell$ will be used to denote its length. Specifically, the length is the \textit{number of edges crossed by the contour}. A domain of frustration (relative to the majority), or ``droplet", needs to grow to have at least macroscopic length $\ell = O(L^{\alpha})$ with energy barrier is $O(J L^{\alpha})$. This is the lowest energy barrier of an error pathway. The entropy of a contour is lower bounded as
\begin{equation}
S(\ell) \geq \log(\Omega(\ell)) \geq \log(L^{(1-\alpha)\ell}2^\ell).
\end{equation}
Here $L^{1-\alpha}$ results from the fact that, for each edge, there are $L^{1-\alpha}$ possible edges crossed by the contour. The $2^\ell$ term lower bounds the entropy originating from the shape of the contour as a random walk. We estimate the rescaled time $\tau$ based on the Arrhenius law for a reaction rate as
\begin{equation}
\tau = \frac{\exp(\beta E_\text{b})}{\Omega} = \exp(\beta F_\text{b})
\end{equation}
where $E_\text{b}$ is the energy barrier associated with the domain and $F_\text{b}$ is the corresponding \textit{free energy} barrier. In this calculation to find an upper bound on $\tau$, we have corrected the Arrhenius law to account for the different paths by which this macroscopic energy barrier can be overcome:
\begin{equation}
F_\text{b}(\beta, L) \leq 2JL^{\alpha}-\frac{1}{\beta}\log\left(L^{(1-\alpha) L^{\alpha}}2^{L^{\alpha}}\right).
\end{equation}
Observe that at a given $\beta$ there is a system size $L_\text{max}$ such that \text{the lower bound on }$F_\text{b}$ is maximized, given by
\begin{equation}
L_\text{max} \propto \exp(c\beta) \label{eq:Lmax}
\end{equation}
with $c > 0$ and dimensions of energy or temperature. Substituting this into the expression for the energy cost,
\begin{equation}
E_\text{b,max}\propto \exp(c'\beta)
\end{equation}
with $c' > 0$. Finally, substituting it into the expression for the bound on the memory time,
\begin{equation}
\tau \leq \tau_\text{max}\propto \exp(\kappa '  \exp(\kappa\beta)).
\end{equation}
with $\kappa$ and $\kappa '$ positive.

While we relied on the Arrhenius law to estimate an upper bound on quantum memory time, the above argument can be turned into a rigorous derivation of upper bound on quantum memory time and lower bounds on gaps in $\mathcal{L}$. Rigorous approaches of finding mixing time for the Glauber dynamics often relies on refinement of the Peierls argument, see Ref.~\citen{Martinelli} for instance. Namely, one considers canonical paths between two configurations in order to bound the mixing time under Markovian dynamics~\cite{Temme14}. Our estimates of the entropy correspond to the logarithm of the number of possible canonical paths connecting two configurations in different ground state sectors. While a rigorous mathematical proof is beyond the scope of this paper, we believe that our estimation of free energy will lead to rigorous derivation of the mixing time for the Ising model on a sparse graph. 

In Ref.~\citen{Bravyi13}, Bravyi and Haah obtained a general formula for obtaining lower bounds on quantum memory time. Applying their bound to the welded code, we obtain a lower bound of the form $\exp(a\exp(a' \beta))$, which was also predicted by Michnicki. Indeed, the optimization of the number of qubits in Ref.~\citen{Bravyi13} can be viewed as maximizing the free energy barrier at fixed $T$. Since our argument provides an upper bound of the double exponential form, we conclude that quantum memory time of the welded code is doubly exponential in the inverse temperature $\beta$. 

\subsection{Thermal topological order and self-correction}\label{sec:thermaltopo}

The Peierls argument, based on estimation of free energy, allowed us to obtain upper bound on quantum memory time, showing that the welded code is not a self-correcting quantum memory. This analysis relates the notion of self-correcting quantum memory to the notion of topological order at nonzero temperature. In section~\ref{sec:topo4}, we showed that the problem of approximating the Gibbs ensemble $\rho_{\beta}$ boils down to evaluations of free energy associated to excitation droplets. Namely, we divide the system into local regions and bound the probability that there exists some ``quantum error-correction'' that will remove excitations inside local regions by sending excitations to sinks; this probability can be viewed as the rate of successful convergence to the thermal Gibbs ensemble. 

\begin{figure}[htb!]
\centering
\includegraphics[scale=0.45]{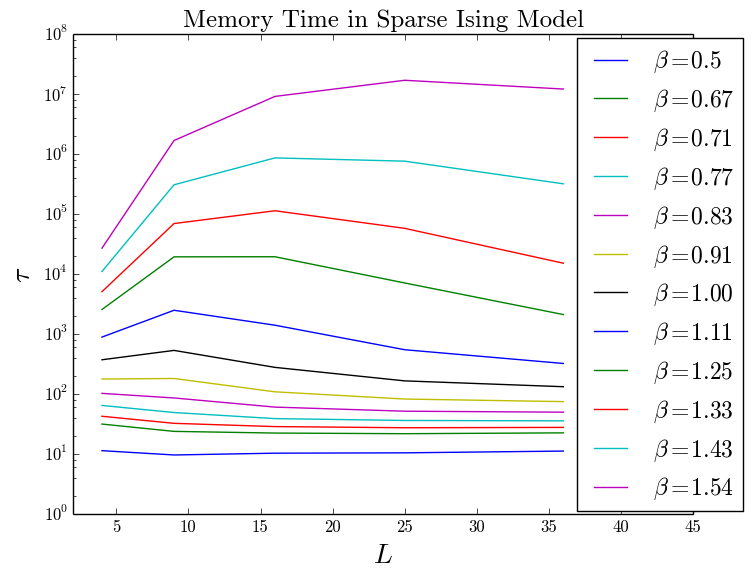}
\caption{Rescaled memory time $\tau$ as a function of system size $L$ at different inverse temperatures. At each temperature, there is an optimal system size, and the temperature at which $L$ is the optimal size goes to zero as $L\rightarrow\infty$.}
\label{fig_data} 
\end{figure}

To be concrete, let us return to the problem of thermal instability of topological order for the welded code. To show that the welded code is not topologically ordered at nonzero temperature, it suffices to prove that an isolated excitation can be eliminated with high probability insider some box of length $R$. We have already seen that the probability of eliminating excitations inside a box is closely related to free energy associated with excitation droplets. In the above Peierls argument, we have observed that free energy becomes negative at $L=L_{0}\approx e^{c'\beta}$ for some positive constant $c'$. This implies that the excitation can be eliminated inside the box of length $R$ if $R\geq L_{0}$. As such, excitations can be eliminated almost surely inside a box of length  $L_{0}\log(L)$ and thus the Gibbs ensemble can be approximated via a range $L_{0}\log(L)$-range circuit. Here the $\log(L)$ factor is necessary because one needs to successfully eliminate excitations inside each grid and there are $\left(\frac{L}{R}\right)^3$ grids.

Having studied memory time of the welded code via the Peierls argument, let us revisit the cubic code. In appendix~\ref{sec:fractal_decomposition}, we provide an ergodic decomposition for quantum fractal codes which enables reduction to problems concerning classical fractal spin models. This approach allows one to obtain upper bound on quantum memory time of quantum fractal codes via the Peierls argument. Here, we derive the upper bound from the perspective of thermal instability of topological order in the cubic code. We have shown that, for a fixed temperature, an isolated excitation can be removed inside a box of length $R \approx \frac{c}{p^2}\log(L)$ for large $L$ where $p=\frac{2}{1+e^{\beta}}$. Let us write $R$ as $R\approx e^{c'\beta}\log(L)$ with some constant $c'>0$, and vary the system size $L$. Let $L_{0}=e^{c''\beta}$ be the value of $L$ such that $R=L$. Then, for $L\leq L_{0}$, an isolated excitation cannot be removed inside some finite box since $R\geq L$, and its elimination requires a quantum circuit of $O(L)$-range. We sketch the growth of complexity of the Gibbs ensemble, as measured by the range of quantum circuits, in Fig.~\ref{fig_complexity} where the growth is linear for for $L\leq L_{0}$ while it suffers from logarithmic slowdown for $L\geq L_{0}$. This reasoning also shows that in the thermodynamic limit the transition temperature $T_c$ to a topologically ordered phase goes to 0.

It is also illuminating to assess the free energy barrier separating different ground state sectors of the cubic code. Observe that $R\leq L$ implies that entropic contributions dominate energetic contributions and the free energy becomes very small. For $R\geq L$, energetic contributions are dominant which is order of $\log(L)$. Thus, the free energy barrier will grow logarithmically and then become small as shown in Fig.~\ref{fig_complexity}. A peak in the free energy barrier is expected to be at $L\approx L_{0}$ which leads to the peak value of $O(\beta)$. This leads to the quantum memory time of $\exp(\kappa\beta^2)$ for some positive constant $\kappa>0$.

\begin{figure}[htb!]
\centering
\includegraphics[scale=0.35]{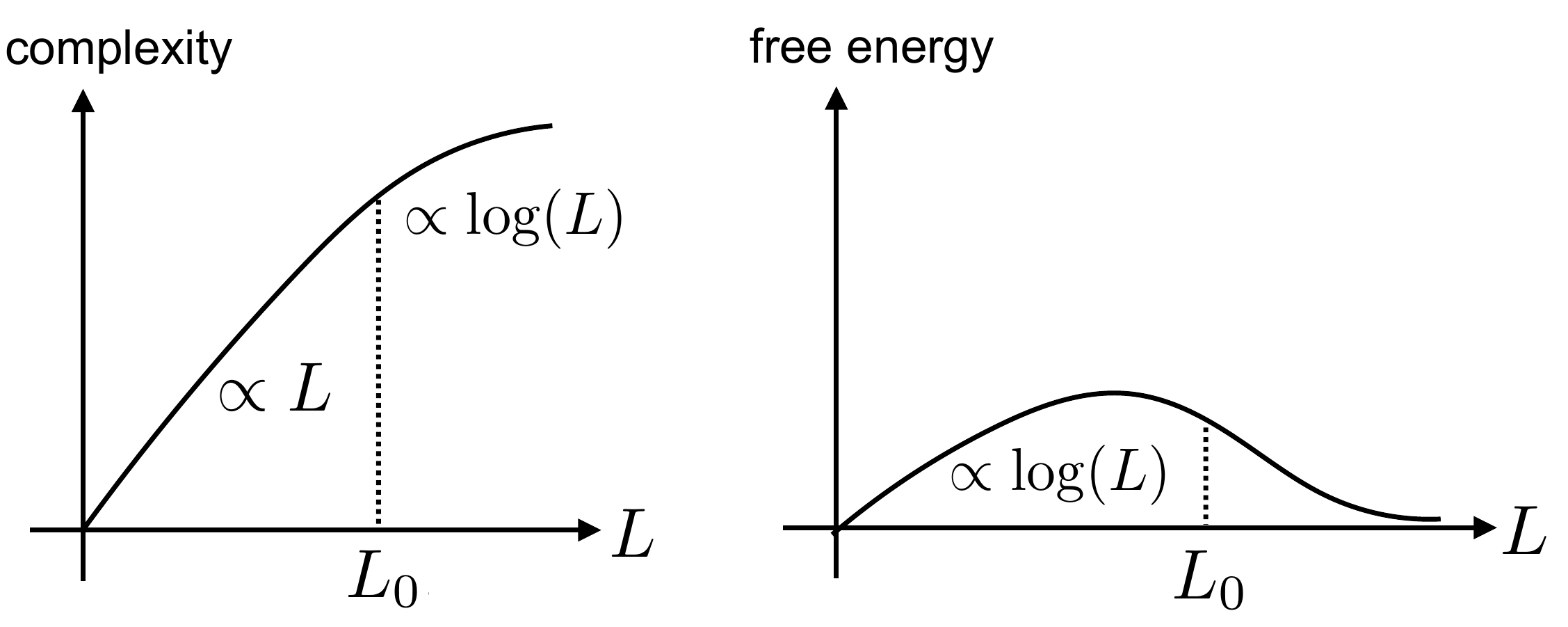}
\caption{Complexity of the Gibbs ensemble and free energy as a function of the system size $L$.}
\label{fig_complexity} 
\end{figure}

\subsection{Computational Simulation}

To verify the analytical result, the classical ferromagnetic sparse Ising model was simulated computationally. The Hamiltonian is $H = - \sum_{\langle i,j \rangle} Z_{i}Z_{j}$ where the linear size of the sparse graph is $L$ while inner squares have linear length of $\sqrt{L}$ ($\alpha = 1/2$). Numerical simulations were performed for values of $L$ between $4$ and $36$, at $\beta^{-1}$ from $0.65$ to $2$. In these numerics, the system is initialized in a completely ordered state, and transitions between successive states were governed by the Metropolis rule. The memory time is defined as the time required for the information stored initially to be compromised by at least 50\%,

As discussed in the previous section, the memory time is expected to be doubly exponential in the inverse temperature which can be very long and poses difficulty in numerical simulations. For this reason, we relied on Kinetic Monte Carlo methods using the Bortz-Kalos-Lebowitz (BKL) algorithm\cite{BKL75}, a variation on the standard Metropolis algorithm. In this scheme, at each step the sites are divided into equivalence classes defined by the energy change that would be caused by flipping the spin at that site. The energy costs are calculated \textit{a priori}, as well as the probabilities of once having flipped such a site that the new configuration is accepted. The BKL algorithm is particularly useful when the rates of spin flips are small.  

The data from these simulations are reported in Fig.~\ref{fig_data}. At each temperature there is an optimal system size, as conjectured by Michnicki and derived in Eq.~(\ref{eq:Lmax}). Thus the welded code is only marginally self-correcting, as at no $T > 0$ can one grow the system arbitrarily large to achieve longer memory times. Qualitatively, this is because in large systems, it is more entropically favorable for droplets of frustration to proliferate and destroy the order. The other major crucial feature, though, is the scaling of the memory time with the inverse temperature $\beta$ for each $L_\text{max}$. Because $L$ is not dense in the numerics, though, one cannot obtain to arbitrary precision the locations of the maxima at each $\beta$. The maxima were instead estimated by the highest data point available from each temperature. Nevertheless, it appears in Fig.~\ref{fig_time-vs-temp} that these data points fit a doubly exponential law well (with coefficient of determination $R^2 = 0.999567$) with respect to $\beta$. The likely implication is that the memory time is not so strongly peaked that an off-peak estimate of the peak is not off by a large factor.

\begin{figure}[htb!]
\centering
\includegraphics[scale=0.45]{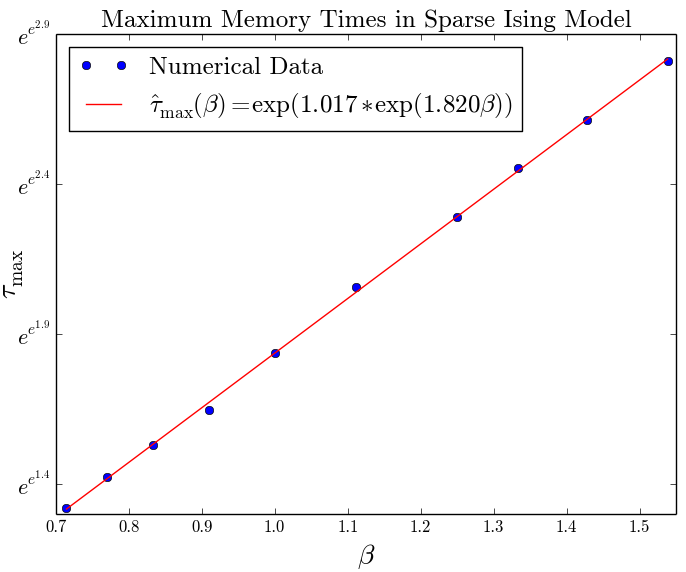}
\caption{$\tau_\text{max}$ as a function of $\beta$, where $\tau_\text{max}$ is the maximum rescaled memory time achieved at each temperature for the classical ferromagnetic sparse Ising model. Coefficient of determination $R^2=0.999567$ for this fit, which suggests that the double exponential prediction is accurate.}
\label{fig_time-vs-temp} 
\end{figure}

\section{Discussion}

In this paper, we proved the thermal instability of the cubic code and studied quantum memory properties of the welded code. In both arguments, we have utilized some variants of free energy arguments, further establishing the connection between self-correction and topological order at nonzero temperature.

Our result raises a fundamental question concerning non-triviality of three-dimensional thermal Gibbs ensemble $\rho(\beta)$, or in other words, whether one can create an arbitrary thermal ensemble efficiently in three dimensions or not. The potential for CSS codes to achieve self-correction or topological order is particularly interesting, as these appear to be closely linked to models in classical physics. Perhaps, a reasonable step toward this question is to ask whether Brell's code is topologically ordered at nonzero temperature or not. The code can be viewed as a four-dimensional toric code supported on a fractal lattice with Haussdorf dimension less than three. It has been argued that the code can be embedded into a three-dimensional space and the code has a finite transition temperature which is a strong indicative of the presence of some ``order'' at nonzero temperature. However such an embedding typically bring pairs of far-away sites to geometric proximity. 

Another interesting question concerns the presence/absence of topological order in an intermediate Coulomb-like phase. Namely, for the four-dimensional qudit toric code, there exists an intermediate temperature range where quantum memory time diverges only polynomially with respect to the system size as opposed to exponential divergence at low temperature phase~\cite{Beni14}. Whether such a phase is topologically ordered or not is an interesting question. 

Our treatment, as well as the original result by Hastings, applies only to commuting frustration-free Hamiltonians. This limitation is due to technical difficulties of studying thermal ensembles of generic gapped Hamiltonians where perturbative analysis is not readily applicable in general. For instance, whether the four-dimensional toric code with small perturbation is topologically ordered at nonzero temperature is currently not known. This may be an interesting future problem.

An imperfect Hamiltonian may be interesting by its own. As a side result, we showed that an imperfect Hamiltonian $H(p)$ of a three-dimensional Hamiltonian $H$ can be written in the form of ``lattice gauge theories'' where terms in $H(p)$ are either vertex and plaquette types after appropriate coarse-graining in appendix~\ref{sec:LGT}. Since lattice gauge theories often admit efficient characterizations via tensor product state (TPS) representations, our finding may hint that thermal Gibbs states for three-dimensional gapped spin systems can be efficiently studied by TPS formalism. 

Generalization of an ergodic decomposition of the Davies generator to arbitrary stabilizer codes is an interesting future problem. Namely, studies of such decompositions in the presence of geometric locality of stabilizer generators may further strengthens no-go results for self-correcting quantum memory in two dimensions.

\section*{Acknowledgment}

We would like to thank John Preskill and Kristan Temme for discussions and Olexei Motrunich for providing computational resources for numerics. This work is supported through Summer Undergraduate Research Fellowship (SURF) and NSF Grant PHY-1125565 at the California Institute of Technology. 
Research at Perimeter Institute is supported by the Government of Canada through Industry Canada and by the Province of Ontario through the Ministry of Research and Innovation.

\appendix

\section{Decomposition for fractals}\label{sec:fractal_decomposition}

The ergodic decomposition method applies to quantum fractal codes, and the problem is reduced to finding classical memory time of classical fractal spin models~\cite{Newman99, Beni11b}. Consider first the two-dimensional toric code represented by using polynomials: $\alpha=1+x$ and $\beta=1+y$. We show the decomposition in Fig.~\ref{fig_decomposition_toric} where squares represent lattice sites which consist of two qubits. (Qubits live on faces, not on edges). To construct the ergodic decomposition for quantum fractal codes, we repeat this decomposition in the $\hat{x}$ direction. That is, for each layer of the $\hat{y},\hat{z}$ plane, we use the decomposition shown in Fig.~\ref{fig_decomposition_toric}. This reduces the problem to that of finding a classical memory time of a classical fractal spin model supported on multiple layers of two-dimensional planes which are attached by forming a comb-like geometry. 

\begin{figure}[htb!]
\centering
\includegraphics[scale=0.23]{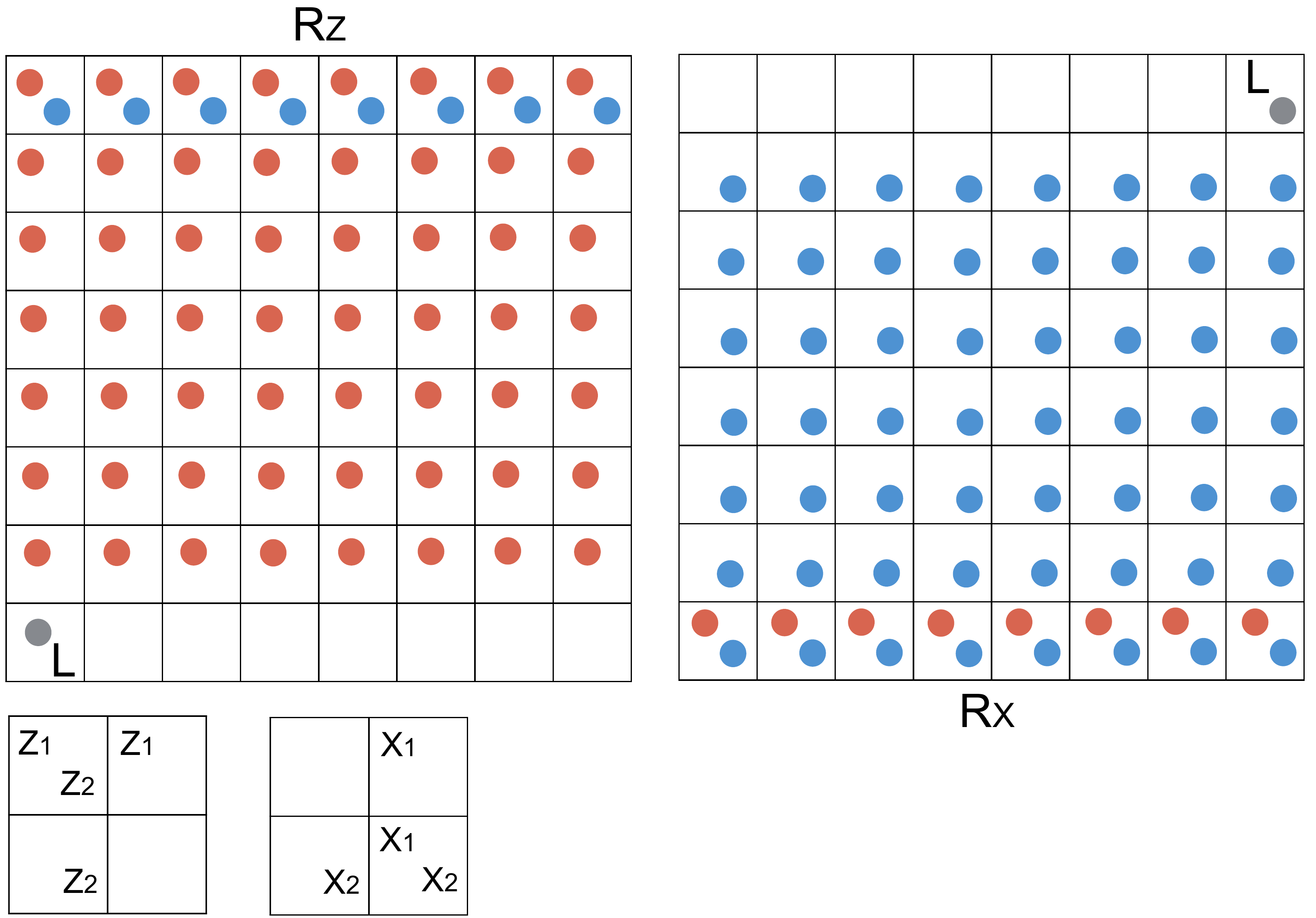}
\caption{Ergodic decomposition for the toric code in a polynomial representation.
} 
\label{fig_decomposition_toric}
\end{figure}

\section{Welding and gapped boundary}\label{sec:boundary}

Welding is a procedure of combining multiple CSS codes into a larger CSS code. This procedure has an interesting interpretation in terms of anyon condensations in gapped boundaries. To be concrete, consider a welding of three copies of the two-dimensional toric code along a welding surface as shown in Fig~\ref{fig_welding_boundary}. Here, ``welding'' refers to a certain procedure of modifying stabilizer generators on the welding surface so that the resulting code has a set of commuting stabilizer generators. 

\begin{figure}[htb!]
\centering
\includegraphics[scale=0.32]{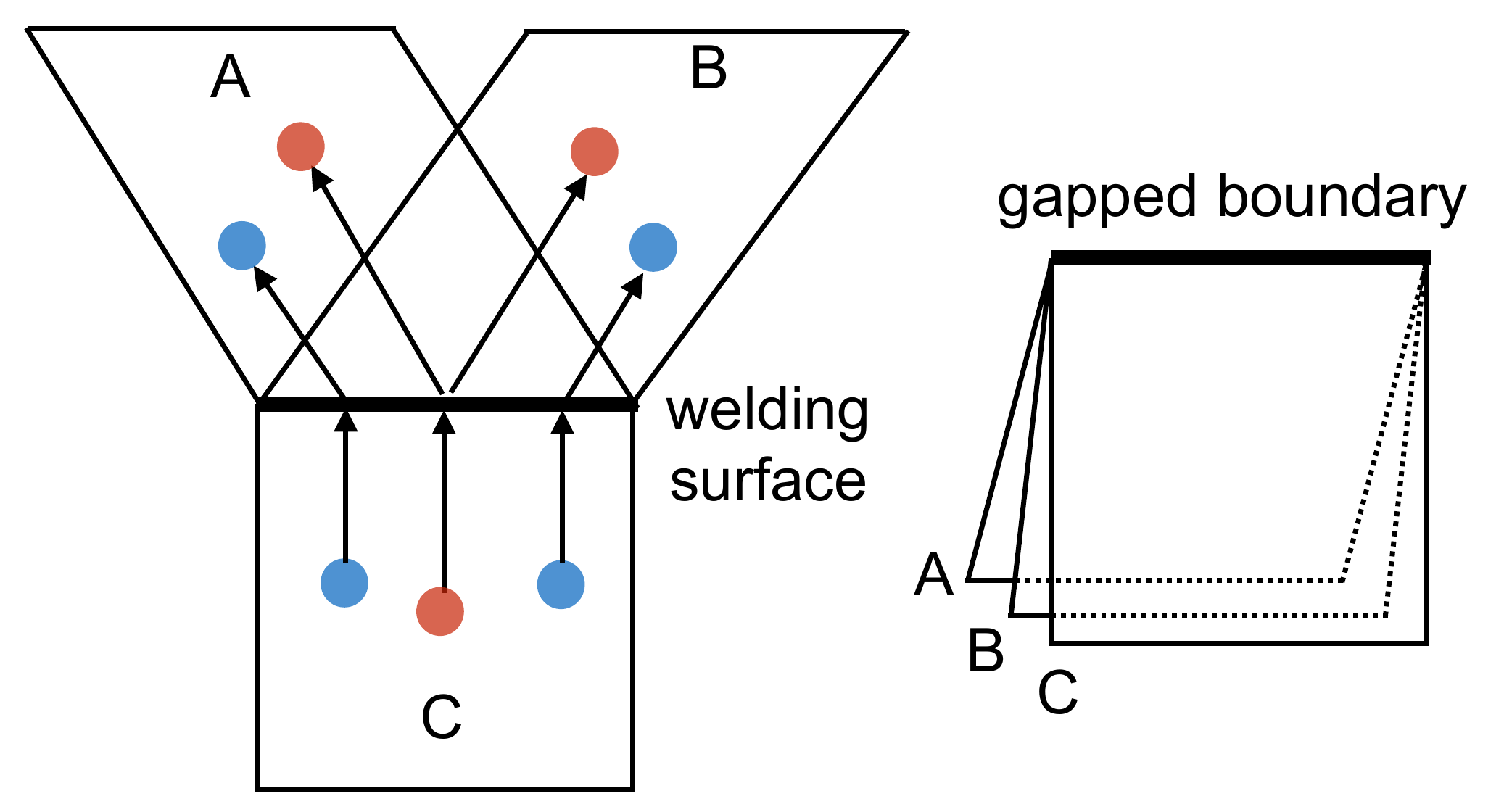}
\caption{Welding as a gapped boundary. 
} 
\label{fig_welding_boundary}
\end{figure}

There are two types of welding procedures, called $X$-type and $Z$-type weldings, which depends on how $X$-type and $Z$-type stabilizer generators are modified. In Fig~\ref{fig_welding_boundary}, we choose the $X$-type welding whose physical meaning is as follows. Let us denote three copies of the toric code by $A,B,C$, and associated charges and fluxes by $e_{A},e_{B},e_{C},m_{A},m_{B},m_{C}$. In this setting, an electric charge from $C$ bifurcates into two electric charges living on $A$ and $B$ respectively. On the other hand, a magnetic flux from each copy of the toric code can travel freely among $A,B,C$. More rigorously, an $X$-type weld of two CSS codes $S_1$ and $S_2$ is constructed as follows: pairs of qubits between $S_1$ and $S_2$ are identified and are each contracted to a single qubit and the new stabilizer group $S$ of the welded code is generated by all $Z$-type operators from $S_1$ and $S_2$ and all $X$-type operators that commute with all of those $Z$-type operators. Further remarks on the codespace can be found in Ref.~\citen{Michnicki14}. By repeating this procedure, one can construct a larger CSS code, which is like a web of two-dimensional toric code. Since electric charges must bifurcate at welding surfaces, one can introduce an energy barrier for propagations of electric charges. Michnicki has applied this trick to multiple copies of the three-dimensional toric code to obtain a quantum memory with polynomially diverging energy barrier. 

To discuss the welding procedure in the language of anyon condensations, it is convenient to fold the entire system along the welding surface as shown in Fig.~\ref{fig_welding_boundary}. Then, the welding surface can be viewed as a gapped boundary defined for three copies of the toric code $A,B,C$. In this interpretation, condensing anyons are
\begin{align}
e_{A}e_{B}e_{C},\quad m_{A}m_{B},\quad m_{B}m_{C}. \label{eq:condensation}
\end{align}
Fig.~\ref{fig_welding_app}(a) shows the terms near the welding surface explicitly. Note that the ergodic decomposition works in a similar manner as shown in Fig.~\ref{fig_welding_app}(b). It has been shown that the three-dimensional topological color code with boundaries is equivalent to three copies of the toric code which are attached by a common boundary whose anyon condensations are given by Eq.~(\ref{eq:condensation})~\cite{Kubica15b}. In other words, the three-dimensional color code can be viewed as an outcome of welding three copies of the three-dimensional color code. 

\begin{figure}[htb!]
\centering
\includegraphics[scale=0.26]{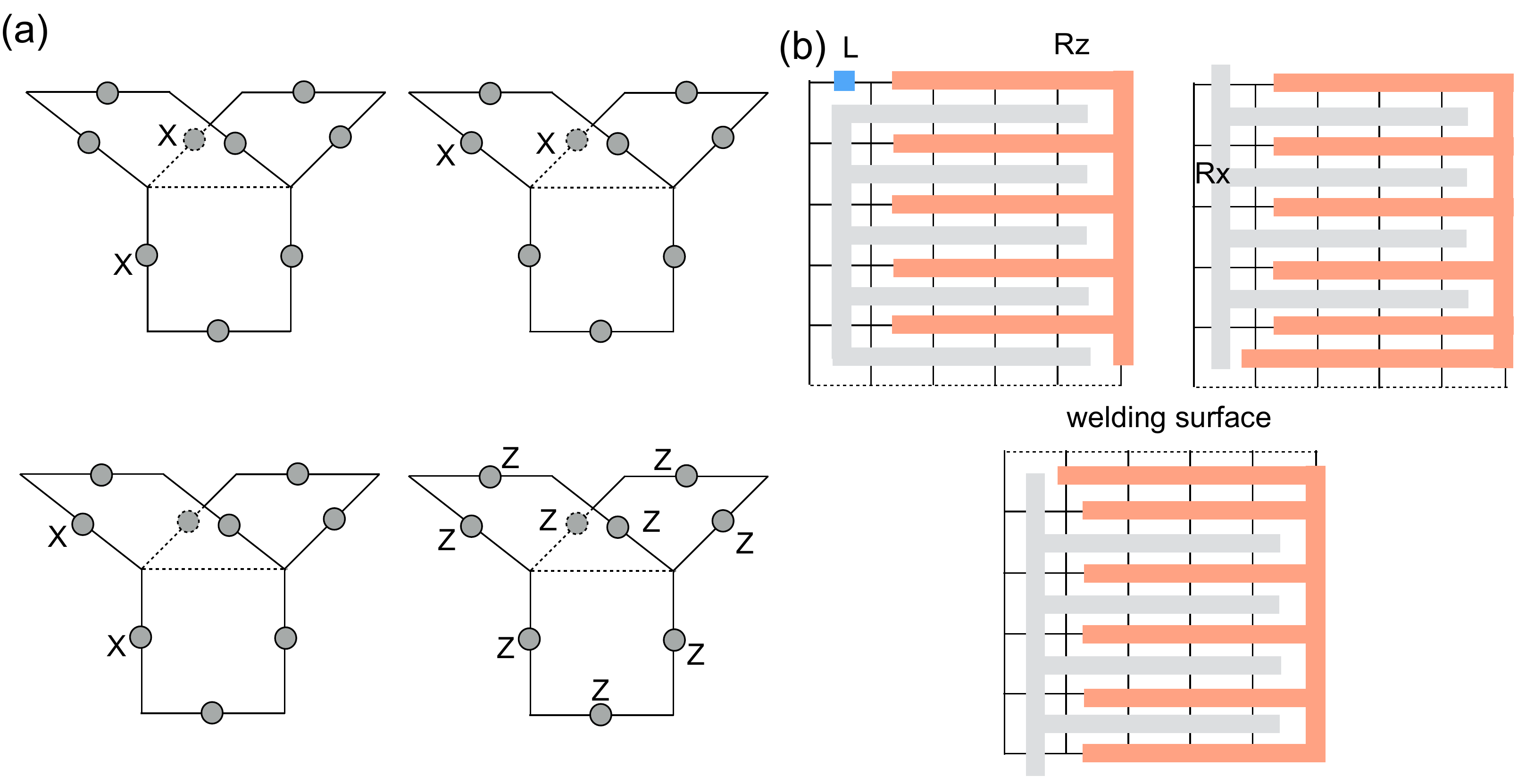}
\caption{(a) Terms near the welding surface. (b) The ergodic decomposition for the 2D welded code.
} 
\label{fig_welding_app}
\end{figure}

\section{More on welded code}

In this appendix, we sketch a proof that the welded code is not topologically ordered at nonzero temperature. A formal construction of the welded code is rather involved, so interested readers are encouraged to read the original paper. The proof can be obtained straightforwardly by examining the patterns of anyon propagations, which was reviewed in appendix~\ref{sec:boundary} from the perspectives of gapped domain walls. 

We begin by showing that the three-dimensional toric code is not topologically ordered at nonzero temperature. The proof was already given by Hastings. The model is defined as follows:
\begin{align}
H = - \sum_{v}A_{v} - \sum_{p}B_{p}
\end{align}
where $A_{v}$ is a Pauli-$X$ vertex term and $B_{p}$ is a Pauli-$Z$ plaquette term. Violations of $A_{v}$ correspond to point-like excitations, and are referred to as electric charges. Violations of $B_{p}$ correspond to loop-like excitations, and are referred to as magnetic fluxes. We then consider an imperfect Hamiltonian $H(p)$ by removing terms of the original Hamiltonian. One technical subtlety is that lemma~\ref{lemma_remove} does not apply to the three-dimensional toric code since there are too many redundant $B_{p}$ operators. However, one can easily modify the lemma by using the CSS property of the toric code. As such, it suffices to consider an imperfect Hamiltonian $H(p)$ where only the vertex terms $A_{v}$ are removed with probability $p$. We split the whole system into cubes of size $O(\log(L)^{1/3})$ which ensures that there is at least one sink of electric charge per cube with probability approaching to unity polynomially. Thus, the three-dimensional toric code is not topologically ordered at nonzero temperature. One can also prove that the three-dimensional toric code is not a self-correcting quantum memory by using the ergodic decomposition. Namely, Fig.~\ref{fig_3D_toric_decomposition} depicts the decomposition where Pauli $Z$ operators (creating electric charges) form multiple comb-like objects attached along a line, which have a finite energy barrier and are essentially described by one-dimensional Ising model, while Pauli $X$ operators (creating magnetic fluxes) form sheet-like objects which are attached like a comb, which have $O(L)$ energy barrier. 

\begin{figure}[htb!]
\centering
\includegraphics[scale=0.4]{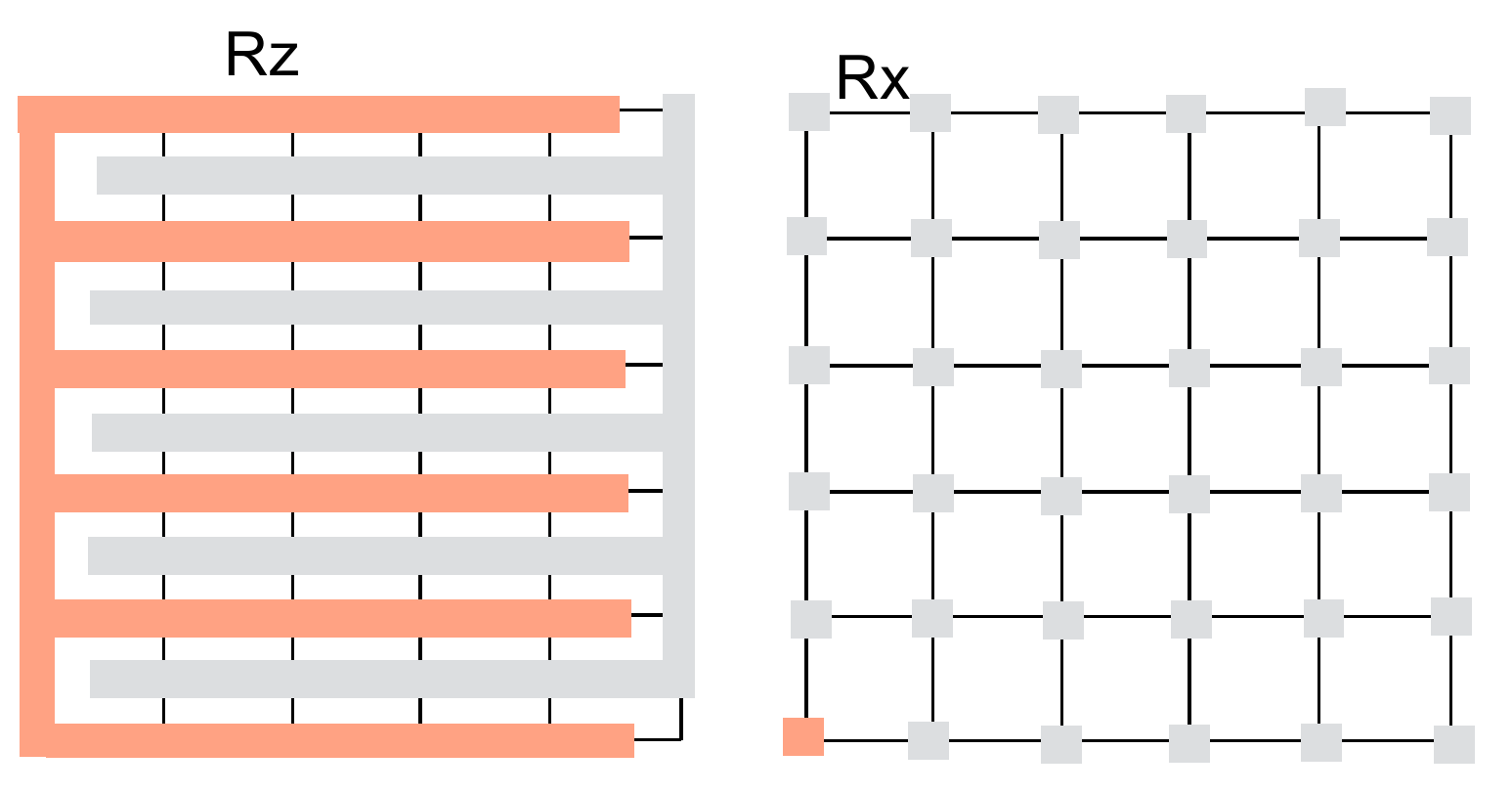}
\caption{The ergodic decomposition for the three-dimensional toric code. The figure shows only two layers of $(x,y)$-planes in the three-dimensional lattice. For the full decomposition, one needs to stack these layers in the $z$ direction.
} 
\label{fig_3D_toric_decomposition}
\end{figure}

The three-dimensional welded code can be constructed by gluing multiple copies of the three-dimensional toric code. As explained in the previous appendix, the wending procedure can be better understood from the paradigm of anyon condensations~\cite{Levin12}. Consider six copies of the three-dimensional toric code. We label their anyonic excitations by $e_{1},\ldots,e_{6}$ and $m_{1},\ldots,m_{6}$ where $e_{j}$ represent electric point-like charges and $m_{j}$ represent magnetic loop-like fluxes. Let us label three coordinates of the lattice by $x,y,z$. We choose boundaries for surfaces perpendicular to $x$ or $y$ such that magnetic fluxes condense. For the boundaries perpendicular to $z$, we construct joint boundaries involving six copies of the toric code. Namely, one can construct a gapped boundary for these six copies of the toric code such that the following anyonic excitations can condense into the boundary:
\begin{align}
(e_{1}\cdots e_{6}), \qquad (m_{1}m_{j}) \quad j=2,\ldots,6.
\end{align}
One can easily see that braiding statistics of condensing anyons are mutually bosonic.  Once the labels of condensing anyons are identified, it is easy to construct the corresponding gapped boundary explicitly, see~\cite{Beigi11, Beni15c} for instance. Note that a single electric charge, say $e_{1}$, cannot condense into the boundary. Instead, if $e_{1}$ moves to the boundary, then it will be reflected back as a composite $e_{2}e_{3}e_{4}e_{5}e_{6}$. Since an electric charge split into multiple excitations, it adds extra energy penalty for propagations of electric charges which are otherwise freely propagating. This joint boundary, involving six copies of the toric code, is indeed the welding surface which glue six copies of the toric code. 

The three-dimensional welded code is essentially a web of multiple copies of the three-dimensional toric code which are glued together with boundaries described above. To be more specific, consider the three-dimensional cubic (sparse) lattice where its linear length is $O(L)$ and its edges have length $O(L^{a})$ ($0<a<1$). (So, there are $O(L^{1-a})$ vertices in one direction of the sparse lattice). We then place the three-dimensional toric code of linear size $O(L^{a})$ on each edge of the lattice, and glue six copies of the toric code at boundaries which are facing to the same vertex of the sparse lattice. The gluing process requires deforming the shape of each of the toric code, but the point is that this process preserves locality of interaction terms as well as keeps the density of qubits finite~\cite{Michnicki14}. By performing such gluing at each vertex, we obtain the three-dimensional welded code. The energy barrier for a magnetic loop-like flux is $O(L^{a})$ which is determined by the barrier for the original toric code of size $O(L^a)$. The energy barrier for an electric charge is $O(L^{2(1-a)})$ since electric charges split into multiple charges once they cross the welding surfaces. By setting $a=2/3$, both excitations will have energy barrier $O(L^{2/3})$.

Let us finally show that the welded code is not topologically ordered at nonzero temperature. We shall split the system into blocks of $O(\log(L))$ qubits. If the box does not contain the welding surface, then all the excitations can be brought to the sink, and can be eliminated with high probability. When the box contains the welding surface, an extra care is necessary. In such cases, we consider six blocks, attached at the welding surface, as a single region. Then, each of six types of excitations from six copies of the toric code can be eliminated in each block. Therefore, all the excitations can be eliminated in a box of size $O(\log(L))$. We also make a brief comment on the ergodic decomposition in the three-dimensional welded code. On the bulk, away from the welding surfaces, the decompositions look like those for the three-dimensional toric code, as depicted in Fig.~\ref{fig_3D_toric_decomposition} while, on the welding surfaces, Pauli $X$ and $Z$ operators are connected to other cubes in a way similar to Fig.~\ref{fig_welding_app}(b). Here we are interested in the decomposition for Pauli-$Z$ operators. Away from the welding surface, Pauli-$Z$ operators form a comb-like object whose thermodynamic properties can be essentially studied as a one-dimensional Ising model. On the welding surface, their trajectory splits into multiple combs which gives energy penalty for electric charges. As a whole, Pauli-$Z$ operators form an object which is like a sparse Ising model.

\section{Reduction to lattice gauge theory}\label{sec:LGT}

In this section, we show that an imperfect realization $H(p)$ of a three-dimensional local Hamiltonian $H$ can be almost surely coarse-grained into the form of lattice gauge theories. By lattice gauge theories, we mean a class of interacting quantum spin systems where physical degrees of freedom live on edges of a graph and interaction terms are either vertex or plaquette types:
\begin{align}
H_{LGT} = - \sum_{v} \sum_{j} {A^{(v)}}_{j} - \sum_{p} \sum_{j} {B^{(p)}}_{j}
\end{align}
where ${A^{(v)}}_{j}$ acts only on edges radiating from a vertex $v$ and ${B^{(p)}}_{j}$ acts on edges that form a plaquette $p$. 

\begin{figure}[htb!]
\centering
\includegraphics[width=0.60\linewidth]{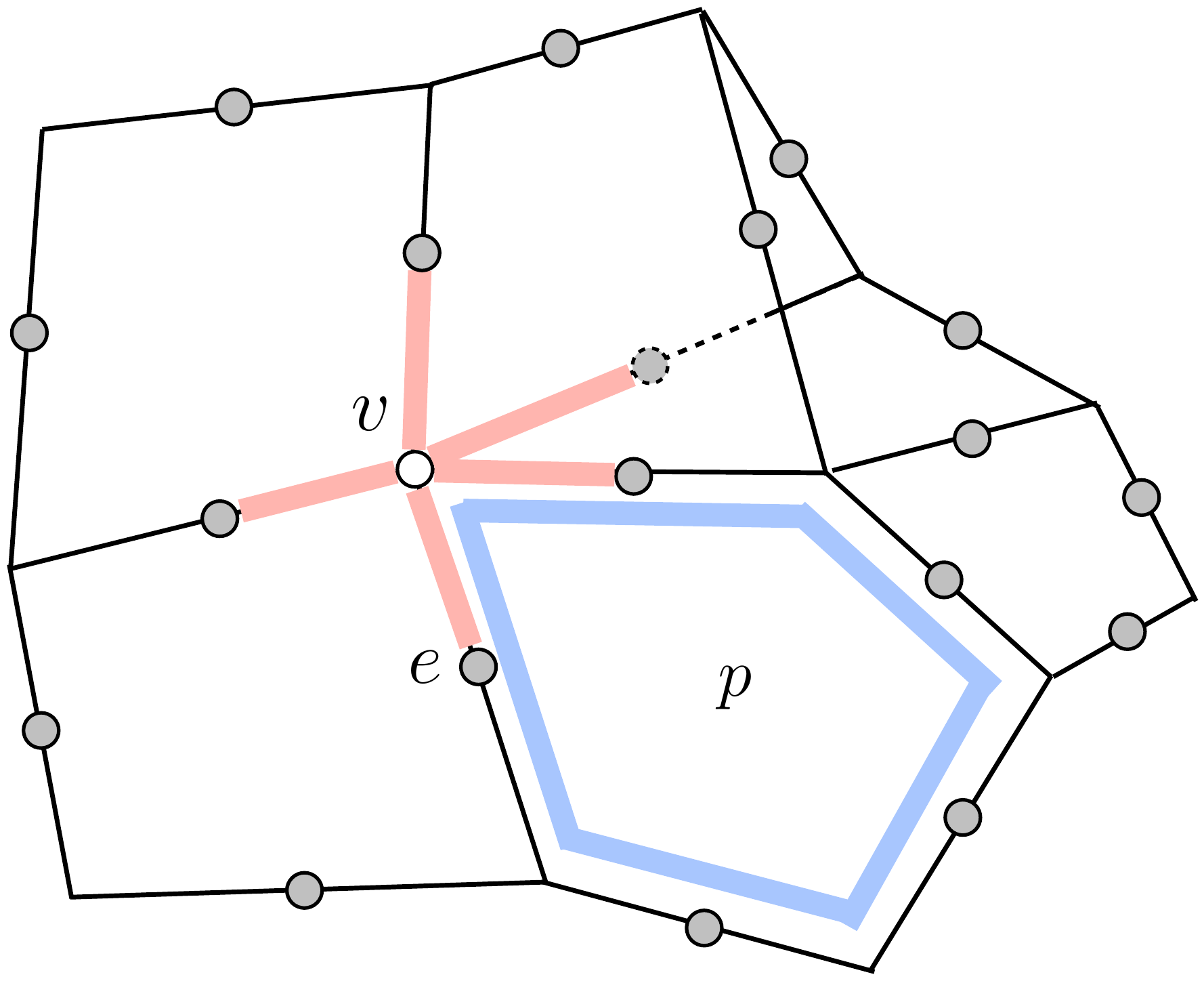}
\caption{A lattice gauge theory $H$ defined on a homogeneous graph. Composite spins live on edges of the graph. Red lines represent a vertex operator on a vertex $v$ and blue lines represent a plaquette operator on a plaquette $p$. 
} 
\label{fig_LGT}
\end{figure}

For a proof, one terminology needs to be introduced. Let $r>0$ be a range of interaction terms, meaning that each term is contained in some ball of radius $r$. By assumption, $r$ is finite. Consider a ball of finite radius $r'$ ($r'\gg r$) such that all the terms in $H(p)$ act trivially on the ball. We call such a ball of no interaction terms an immune sphere. Note that an immune sphere appear with some finite probability as long as $r'$ is finite. 

We begin by demonstrating a reduction to a lattice gauge theory for a certain simple realization of $H(p)$. Consider a cubic grid with spacing $d\gg r'$ on a lattice, denoted by $\Lambda_{grid}$. First, we study a special case where immune sphere are present at vertices of the grid. We draw lines connecting immune spheres and the centers of the cubes, then split the entire system into a collection of octahedral objects (Fig.~\ref{fig_coarse}(a)), each of which can be identified with a face of the cubic lattice. The system is coarse-grained so that degrees of freedom in each octahedral object form single composite spins which live on faces of the cubic grid as in Fig.~\ref{fig_coarse}(b). We then construct a cubic graph $\Lambda$ by connecting centers of unit cells in the grid so that composite particles live on edges of this graph.

\begin{figure}[htb!]
\centering
\includegraphics[width=1.0\linewidth]{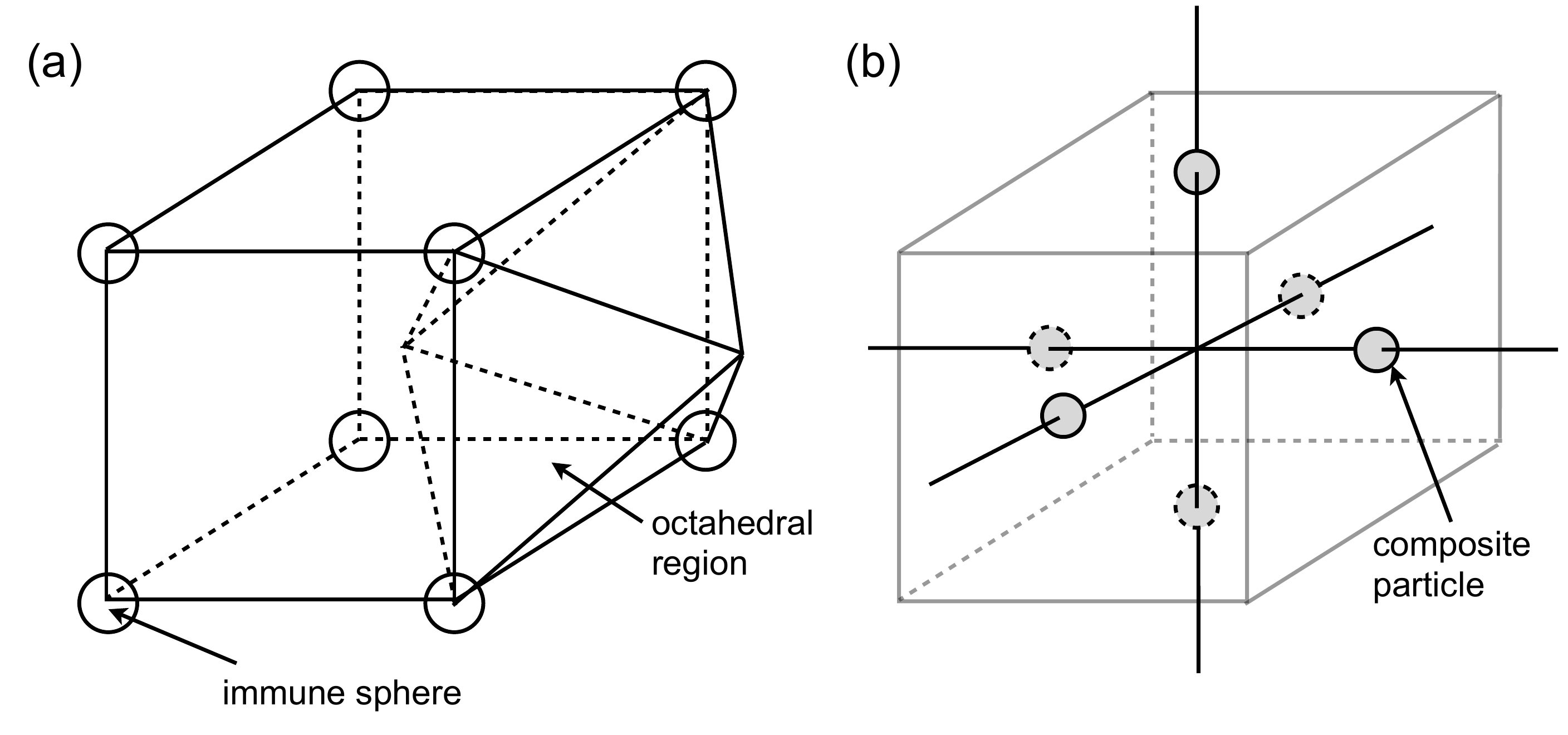}
\caption{(a) Immune spheres at vertices on a cubic grid. Shaded dots represent spheres at half-integer sites. We split the entire system into octahedral objects and view them as composite spins. (b) A construction of a lattice gauge theory. Composite spins live on faces of a cubic lattice. 
} 
\label{fig_coarse}
\end{figure}

We now show that interaction terms of a local Hamiltonian $H$ are either vertex or plaquette type. One can classify interaction terms into three types. (a) Terms which live near an edge of the cubic grid. (b) Terms which live near a face of the cubic grid. (c) Terms which live in the bulk of a cube. For (a), such terms can be considered as plaquette terms which are at most four-body. For (b), such terms are one-body operators. For (c), such terms can be considered as vertex terms which are at most six-body. Thus, the system can be viewed as a lattice gauge theory on a graph $\Lambda$.  

Let us treat a generic case where interaction terms are removed randomly with probability $p$. We split the full three-dimensional system into unit cubes of $d\times d\times d$ qubits where $d>0$ is some sufficiently large positive integer such that $d\gg r' \gg r$. For a given unit cube, the probability of finding at least one immune sphere of radius $r'$ is lower bounded by some finite constant $q>0$. By taking sufficiently large $d$ while keeping $r'$, one can make $q$ arbitrarily close to unity. One can connect immune spheres and construct a skewed cubic grid. We then use a coarse-graining strategy similar to the one used in the previous section. This reduces the system into a form of lattice gauge theories.

%


\end{document}